\documentclass[sigconf,nonacm]{acmart} 

\AtBeginDocument{%
  }

\setcopyright{acmlicensed}
\copyrightyear{2018}
\acmYear{2018}
\acmDOI{XXXXXXX.XXXXXXX}
\acmConference[Conference acronym 'XX]{Make sure to enter the correct
  conference title from your rights confirmation email}{June 03--05,
  2018}{Woodstock, NY}
\acmISBN{978-1-4503-XXXX-X/2018/06}

\usepackage{algpseudocode, algorithm, algorithmicx}
\newcommand*\Let[2]{\State #1 $\gets$ #2}
\algnewcommand\PublicKnowledge{\item[\textbf{Public knowledge:}]}
\algnewcommand\Input{\item[\textbf{Input:}]}
\algnewcommand\Notation{\item[\textbf{Notation:}]}
\algrenewcommand\algorithmicthen{}
\DeclareMathOperator{\sort}{sort}
\DeclareMathOperator{\shuffle}{shuffle}
\newcommand{\share}[1]{[\![#1]\!]}
\newcommand{\osort}[1]{\share{\sort(#1)}}
\newcommand{\oshuffle}[1]{\share{\shuffle(#1)}}
\DeclareMathOperator{\val}{val}
\DeclareMathOperator{\coord}{coord}
\DeclareMathOperator{\minleft}{MinLeft}
\DeclareMathOperator{\maxleft}{MaxLeft}
\DeclareMathOperator{\minright}{MinRight}
\DeclareMathOperator{\maxright}{MaxRight}
\DeclareMathOperator{\nnz}{nnz}

\usepackage{multirow}
\usepackage{subcaption}

\begin{document}

\title[Secure sparse multiplications]{Secure Sparse Matrix Multiplications and their Applications to Privacy-Preserving Machine Learning}

\author{Marc Damie}
\orcid{0000-0002-9484-4460}
\affiliation{%
	\institution{University of Twente}
	\city{Enschede}
	\country{The Netherlands}}

\author{Florian Hahn}
\orcid{0000-0003-4049-5354}
\affiliation{%
	\institution{University of Twente}
	\city{Enschede}
	\country{The Netherlands}}

\author{Andreas Peter}
\orcid{0000-0003-2929-5001}
\affiliation{%
	\institution{University of Oldenburg}
	\city{Oldenburg}
	\country{Germany}}

\author{Jan Ramon}
\orcid{0000-0002-0558-7176}
\affiliation{%
	\institution{Inria}
	\city{Villeneuve-d'Ascq}
	\country{France}}

\renewcommand{\shortauthors}{Damie et al.}

\begin{abstract}
	To preserve data privacy, multi-party computation (MPC) enables executing Machine Learning (ML) algorithms on private data.
	However, MPC frameworks do not include optimized operations on sparse data.
	This absence makes them unsuitable for ML applications involving sparse data; e.g., recommender systems or genomics.
	Even in plaintext, such applications involve high-dimensional sparse data, that cannot be processed without sparsity-related optimizations due to prohibitively large memory requirements.

	Since matrix multiplication is a central building block of ML algorithms, our work proposes dedicated MPC algorithms to multiply secret-shared sparse matrices.
	Our sparse algorithms have several advantages over secure dense matrix multiplications (i.e., the classic multiplication).
	On the one hand, they avoid the memory issues caused by the ``dense'' data representation of dense multiplications.
	On the other hand, our algorithms can significantly reduce communication costs (up to $\times1000$) for realistic problem sizes.
	We validate our algorithms in two machine learning applications where dense matrix multiplications are impractical.
	Finally, we take inspiration from real-world sparse data properties to build 3 techniques minimizing the public knowledge necessary to secure sparse algorithms.
\end{abstract}




\maketitle
\section{Introduction}
\label{sec:intro}
MPC protocols are cryptographic primitives allowing to compute functions on secret inputs.
They have been applied in many contexts involving privacy concerns such as privacy-preserving ML (PPML) \cite{mohassel_secureml_2017} or telemetry \cite{asharov_efficient_2022}.
However, existing protocols remain impractical for many real-world applications, especially when the applications involve high-dimensional \emph{sparse data}, i.e., data with a large majority of zero values.

Sparse data arises naturally in many applications including recommender systems, genomics, and natural language processing.
For instance, in a video recommender system, each user watches only a tiny fraction of the catalog.
We can encode a user’s history as a binary vector, with 1 indicating a watched video and 0 otherwise, yielding very long vectors that are mostly zeros.
Real-world datasets can be extremely sparse \cite{li_convergence_2017}: about 99\% zeros in Netflix dataset, 99.9\% in Tox21, 99.99\% in Bookcrossing, and 99.999\% in Flickr.
At such sparsity rates, specialized algorithms are essential; otherwise, computations become inefficient, or even infeasible if the dataset cannot fit in memory using dense storage.

Storing large sparse datasets in a dense format (one value per cell) demands excessive memory.
On real-world sparse data, this ``dense storage'' often becomes prohibitive; making computation infeasible due to memory constraints.
Moreover, dense representations lead to inefficient linear algebra operations, as most of the computational effort is wasted processing zeros.
Over the last 50 years, algorithms to process plaintext sparse data efficiently have been proposed \cite{buluc_implementing_2011,duff_survey_1977}.
Since 2002, dedicated software libraries \cite{duff_overview_2002} popularized them.
Sparse data is present in many ML applications but also in other fields, including graph algorithms.
Our paper focuses on PPML applications, but our results apply to any application involving private sparse data.

Several MPC frameworks enable training ML models on secret-shared data \cite{koti_swift_2021,lu_aegis_2023,mohassel_aby3_2018,mohassel_secureml_2017,patra_blaze_2020,wagh_securenn_2019,wagh_falcon_2021}.
However, \textbf{none of these frameworks provide algorithms optimized for sparse data}.

The scalability limitations (especially memory issues) of dense plaintext algorithms transfer naturally to dense MPC algorithms.
Thus,  \textbf{MPC frameworks need dedicated sparse algorithms} for ML applications like recommender systems.
These frameworks are not fundamentally incompatible with sparse data; there is simply no suitable secure sparse algorithm in the literature.

\paragraph{Related works}
There exist protocols to sum sparse vectors securely, including sorting-based \cite{asharov_efficient_2022,bell_distributed_2022}, and shuffling-based \cite{xu_harnessing_2025}.

A few papers \cite{chen_when_2021,cui_exploiting_2021,schoppmann_make_2019} studied secure sparse multiplications in a two-party setting.
However, these protocols require the data owners to actively participate because one of the computation parties must \emph{know the plaintext sparse data}.
Hence, they do not support more than two data owners simultaneously, which is too limiting, as modern ML applications involve thousands of data owners.

Typical ML applications require a more generic MPC setup: outsourced training.
In this setting, the data owners share their secret data with a group of computation servers and disconnect.
This distinction between input parties and computation parties enables supporting an arbitrary number of data owners, contrary to existing works \cite{chen_when_2021,cui_exploiting_2021,schoppmann_make_2019}.
Appendix \ref{app:adapt-non-outsourced} details why the adaptation of these works either results in an inefficient construction or is not straightforward.
The secure multiplication of sparse matrices in the outsourced setting is then \textbf{a significant open problem}.
\emph{Our work will focus on this outsourced setting; providing sparse multiplications compatible with any MPC framework}.

Some information-theory works \cite{egger_sparse_2023,xhemrishi_distributed_2022} studied the multiplication of a secret-shared matrix with a public matrix.
Due to the public input, their setting differs from classic MPC where all inputs are secret.
Moreover, Yu et al. \cite{yu_lodia_2025} has recently proposed a sparse matrix-vector multiplication based on fully homomorphic encryption (FHE).
While they highlight the importance of secure sparse multiplications, we cannot compared directly to them as their FHE threat model (single dishonest server) differs from our MPC-based one (group of non-colluding servers).
Finally, Nayak et al. \cite{nayak_graphsc_2015} presented secure graph computations (a problem closely related to sparse operations due to the sparse adjacency matrices), but they do not present multiplication protocols.

\paragraph{Our contributions}
We present \textbf{two secure algorithms} to compute matrix multiplications on secret-shared sparse data.
Our algorithms \textbf{avoid the memory issues} of dense matrix multiplications on sparse data (e.g., in an experiment from Sec. \ref{sec:exp_analysis}: 19TB using dense multiplications vs. 60GB using ours).
Moreover, our algorithms \textbf{reduce the communication costs} up to a factor 1000, and we implement two ML applications that are impractical with existing secure algorithms.
Third, as we highlight that any efficient MPC protocol on sparse data (also acknowledged in \cite{schoppmann_make_2019}) requires public knowledge about the sparsity, we leverage the properties of real-world sparse data to \textbf{minimize this mandatory public knowledge} and obtain it in a privacy-preserving manner.

\paragraph{Notations}
Let $\nnz(x)$ (resp. $\nnz(X)$) refer to the number of non-zero values in a vector $x$ (resp. matrix $X$).
Let $\share{a}$ refer to a share of a secret value $a$.
In our algorithms, any operation involving shared values corresponds to a secure variant; e.g., $\share{a}+\share{b}$ is the secure addition of $a$ and $b$ and $\osort{l}$ is the oblivious sort of $l$.
Secure algorithms are executed jointly by all shareholders.

\section{Preliminaries}
\subsection{Threat model}
\label{subsec:threat_model}
Threat modeling describes the expected behavior of the agents involved in an MPC protocol: an \emph{honest} agent follows the protocol, an \emph{honest-but-curious} agent follows the protocol and attempts to infer private information passively, and a \emph{malicious} agent does not necessarily follow the protocol and can actively disrupt the protocol to obtain private information.
MPC works distinguish security against a majority of honest agents from security against a majority of dishonest agents.
``Dishonest agent'' is an umbrella term encompassing malicious and honest-but-curious agents (depending on the security setup considered).

We analyze our algorithms under honest and dishonest majority.
We focus on honest-but-curious security, but our algorithms rely on oblivious shuffling and sorting with known maliciously secure variants \cite{hamada_oblivious_2014,laur_round-efficient_2011}.
We leave for future work the extension to malicious security and the subsequent security proofs.

In line with related works, our threat model excludes poisoning attacks where malicious data owners forge input data so the protocol reveal private information.
We assume the data owners provide well-formed secret-shared inputs.

\subsection{Outsourced setting}
Our contribution focuses on an ``\emph{outsourced}'' computation setup (popular in recent PPML works \cite{hegde_sok_2021,mohassel_secureml_2017,wagh_securenn_2019,wagh_falcon_2021}): each data owner shares their secret data with a group of computation servers.
These servers build a secret-shared matrix $\share{X}$ containing all data owners' inputs; e.g., one row per data owner -- referred to as ``horizontal'' data partitioning in PPML literature.
Then, they can perform secure computations on this matrix.
For example, the servers can execute a protocol to train an ML model $\share{\theta}$ on the dataset $X$.
In this setting, the data owners share their secret data and disconnect.
Outsourced protocols support an arbitrary number of data owners because they separate input parties (i.e., data owners) from computation parties.


Contrary to many PPML works that support a two-party \cite{mohassel_secureml_2017,hegde_sok_2021} or three-party \cite{wagh_securenn_2019,wagh_falcon_2021} outsourced setups, our work support an unbounded number of computation parties.
Nevertheless, secret-sharing protocols induce a quadratic communication cost in the number of parties.
This quadratic cost makes outsourcing necessary to support thousands of data owners like in ML applications.

All existing secure sparse multiplications \cite{chen_when_2021,cui_exploiting_2021,schoppmann_make_2019} support a non-outsourced setting.
They presented \emph{two-party protocols} assuming \emph{the computation parties are also data owners} ($\Rightarrow$ two data owners maximum) because they require one of the computation parties to know the plaintext sparse data.
Such a setting is too limiting for ML applications as they typically involve many data owners.

\textbf{Adapting existing protocols to avoid designing new ones?}
The adaptation is not easy because a plaintext knowledge is often essential to these protocols.
\textbf{Appendix \ref{app:adapt-non-outsourced}} details why the adaptation of  \cite{chen_when_2021,cui_exploiting_2021,schoppmann_make_2019} to the outsourced setting either results in an inefficient construction or is not straightforward.

\subsection{MPC primitives}

\textbf{Secret sharing} \cite{shamir_how_1979} enables sharing a secret value $x$ between multiple parties.
Each shareholder $j$ holds a share $\share{x}_j$, revealing \emph{nothing} about $x$.
The shareholders can reconstruct the secret value by gathering their shares.
Several secret sharing schemes exist, e.g., additive shares are finite field elements drawn randomly, s.t. $\sum_j [\![ x ]\!]_j = x$.
Many protocols \cite{evans_pragmatic_2018} enable processing secret-shared values; from arithmetic operations to complex operations (e.g., sorting).

Our algorithms rely notably on \textbf{oblivious shuffling and sorting}.
``Oblivious sorting'' designates MPC protocols for sorting secret-shared values without revealing information about the values.
Sorting networks \cite{hamada_oblivious_2014} are straightforward oblivious sorting solutions, as their execution is data-independent.
Their cost is $O(m\log^2 m)$ for a list of size $m$.
The oblivious radix sort \cite{hamada_oblivious_2014} is the state-of-the-art honest-majority protocol, with its optimal complexity of $O(m\log m)$ (and low constant factors).
Recent papers optimized this protocol in a three-party setting \cite{araki_secure_2021,asharov_efficient_2022}.

``Oblivious shuffling'' protocols randomly permute lists of secret values without revealing the permutation.
Laur et al. \cite{laur_round-efficient_2011} proposed an honest-majority protocol with linear costs, used in many recent protocols \cite{araki_secure_2021,asharov_efficient_2022,hamada_oblivious_2014}.
A recent work \cite{lu_rpm_2023} also described dishonest-majority protocols, but with supra-linear costs.

\subsection{Matrix multiplication}
\label{subsec:mat_mult}

Let $Z$ be the result of the multiplication between the $n\times m$ matrix $X$ and the $m\times p$ matrix $Y$.
$Z$ is an $n\times p$ matrix defined as follows:
\begin{equation}
	\label{eq:mat_mult}
	Z_{ij} = \sum_{k=1}^m x_{ik}\cdot y_{kj}, \forall i\in\{1\ldots n\}, j\in\{1\ldots p\}
\end{equation}

This formula holds for all matrices.
However, the implementation may differ depending on the type of matrices (e.g., vectors, diagonal matrices, or triangular matrices).
Our work focuses on optimized algorithms for sparse matrices.
Sparse algorithms avoid ``useless'' multiplications and only multiply the non-zero values.
Indeed, if the matrix multiplication is implemented naively (i.e., triple nested loop), the algorithm would compute many scalar multiplications with a null result.
Hence, sparse algorithms leverage the matrix sparsity to have a cost inverse proportional to the sparsity.

Software libraries usually implement dense matrix multiplication via a triple nested loop, resulting in an $O(nmp)$ cost.
While this cost has long been considered the lower bound for dense matrix multiplication, several works described algorithms with a sub-cubic cost \cite{coppersmith_matrix_1987}.
Dumas et al. \cite{dumas_secure_2019} described an MPC protocol to execute a matrix multiplication with sub-cubic cost.
These algorithms have constant factors limiting their practical interest (\textit{even in plaintext}), so software libraries favor algorithms with a cubic cost.

\subsection{Dense matrix multiplication on secret-shared data}
\label{subsec:dense_baselines}

The cubic-cost matrix multiplication (i.e., the naive implementation of Equation \eqref{eq:mat_mult}) can easily be imported to MPC.
However, recent works \cite{chen_maliciously_2020,mohassel_aby3_2018,mohassel_secureml_2017,mono_implementing_2023,patra_blaze_2020,wagh_falcon_2021} described more efficient protocols under honest or dishonest majority assumptions.
These works reduce the communication costs, but the computation cost remains asymptotically cubic.
Moreover, the storage costs remain asymptotically equivalent to the plaintext algorithm, so they cannot support larger matrices than the plaintext dense algorithm.

\paragraph{Dishonest majority} Beaver's multiplication triplets are a popular building block to speed up secure multiplications via preprocessing steps \cite{evans_pragmatic_2018}.
Mohassel and Zhang \cite{mohassel_secureml_2017} generalized this concept to matrices: matrix multiplication triplets.
These triplets make the communication costs proportional to the input and output size.
For example, the multiplication of an $n\times m$ matrix with an $m\times p$ matrix requires $O(nm + mp + np)$ communications (instead of $O(nmp)$ with the naive algorithm).
This approach reduces the asymptotic cost for matrix-matrix multiplications, but the communication costs of the vector-vector and matrix-vector multiplications remains asymptotically equal to the naive costs.
Matrix multiplication triplets with malicious security were later described \cite{chen_maliciously_2020,mono_implementing_2023}.

\paragraph{Honest majority} Using the honest majority assumption, several recent works \cite{koti_swift_2021,lu_aegis_2023,mohassel_aby3_2018,patra_blaze_2020,wagh_falcon_2021} designed protocols with a communication cost proportional to the output size.
Hence, a square matrix multiplication only requires $O(n^2)$ communications, and a vector multiplication $O(1)$ communications.
Before these recent protocols, De Hoogh \cite{hoogh_design_2012} had already described a matrix multiplication with equivalent costs based on Shamir's secret sharing (SSS).
The recent protocols \cite{koti_swift_2021,lu_aegis_2023,mohassel_aby3_2018,wagh_falcon_2021} are more complex to understand than the SSS-based algorithm of \cite{hoogh_design_2012} because they rely on share conversion protocols.
Moreover, these protocols only support 3 parties, while the SSS-based solution supports an arbitrary number of parties.
Nevertheless, the SSS-based solution only has honest-but-curious security, while \cite{koti_swift_2021,lu_aegis_2023,patra_blaze_2020,wagh_falcon_2021} support malicious security.

To understand the gain provided by the honest-majority assumption, let us detail the optimized vector multiplication under SSS \cite{hoogh_design_2012}.
In SSS, each secret-shared value is represented using a polynomial of degree $t$.
To multiply two secret-shared scalars in SSS, the parties locally multiply their local shares (obtaining a polynomial of degree $2t$) and then execute a degree reduction protocol (via a ``resharing'' protocol).
For dense vector multiplication, we can delay the degree reduction and add all local multiplication results together.
The addition is possible because all local multiplication results are shares of equal polynomial degree (i.e., $2t$).
Then, a single degree reduction is done on the sum result.

\section{Sparse data in PPML}
\label{sec:sparse-data-ppml}
\subsection{Sparse data representation}
The term ``sparse data'' refers to data containing a large proportion of zero values.
If stored using the default matrix format, the zeros occupy a significant memory space.
When the data contains a sufficient number of zeros, sparse data representations can significantly reduce the storage space.

Several sparse data representations exist \cite{buluc_implementing_2011}, each with specific algorithms to perform operations on the data.
Depending on the data representations, some operations are more efficient than others.
In other words, there is no perfect sparse data representation.

We use the \textbf{tuple representation} \cite{buluc_implementing_2011}: each sparse vector is a list of non-zero tuples: $t = (i, v_i)$ with $i$ the coordinate in the vector and $v_i$ the non-zero value.
Let $\coord(t)=i$ denote the coordinate of the non-zero and $\val(t)=v_i$ its value.
This format is already used in plaintext operations, like in the SciPy library in which it is called COO format (i.e., COOrdinate format).

Other sparse data formats like the Compressed Sparse Row format are also used in plaintext.
However, they require look-up operations which are particularly expensive in MPC.
Thus, our work focused on the tuple representation.

We refer to algorithms supporting sparse data as ``sparse algorithm'', a common abuse of terminology.
We use the term ``dense algorithms'' to refer to classic algorithms.

\subsection{Larger datasets are sparser}
\label{subsec:larger-sparse}
\paragraph{Optimizing dense algorithms, a dead end}
To process sparse data, one may consider to use dense matrix multiplication, with more memory.
Unfortunately, it would require tremendous amount of memory.
Secure dense algorithms have the same asymptotic memory footprint and hence
run into the same memory issue.
There is a limit to the scalability of dense algorithms.
Even if extreme amounts of memory would be available, dense secure multiplications would still fail to scale because their communication cost is typically at least linear in the size of the dense input.

A workaround to avoid this memory issue could be to partition the matrix in sub-matrices and perform dense multiplications on these sub-matrices.
This approach is effective in classic linear algebra to handle large dense datasets and could (to some extent) reduce memory consumption.
However, it would still require storing the product as a dense matrix; the memory benefit will not be significant.
Moreover, it adds significant costs for repeated sparse-to-dense conversions.
Finally, this partitioning approach would not reduce the communication costs of MPC protocols that would remain correlated to the full matrix size.

To sum up, sparse algorithms are popular in plaintext because there is no other efficient solutions to process high-dimensional sparse datasets.
In plaintext, dense matrix multiplications usually have a sparse equivalent (e.g., NumPy for dense and SciPy for sparse in Python).
To cover all ML applications, MPC frameworks should similarly have sparse algorithms.
The first motivation of sparse algorithms is to avoid the memory issues caused by dense algorithms.
Furthermore, they often provide significant performance gains.

\paragraph{Sparsity-size correlation}
Real-world sparse datasets show an interesting phenomenon: the sparsity is correlated to the matrix size.
Indeed, larger sparse datasets have in average a larger sparsity.
Thus, \textbf{the larger datasets are, the more beneficial the sparse algorithms become.}

To understand this phenomenon, let us take the example of recommender systems on a marketplace (e.g., Amazon): the matrix quantifies the interaction of each consumer with each product.
The dataset is sparse because each consumer only interacts with a small subset of all possible products.
If the marketplace increases its number of products by a factor 1000, the consumer will not consume 1000 times more products.
Even if the consumption increases a bit, existing datasets show that it does not follow the matrix growth; i.e., the sparsity increases.
The real-world datasets that we arbitrarily selected from \cite{li_convergence_2017} satisfy the same trend: MovieLens 1M has 95\% sparsity for 1.7K columns, Yahoo Movie has 99.7\% sparsity for 12K columns, Bookcrossing has 99.99\% sparsity for 340K columns, and flicker has 99.999\% sparsity for 500K columns.

\subsection{Public knowledge in secure sparse multiplications}
\label{subsec:prior-knowledge}
Other secure sparse algorithms \cite{chen_when_2021,cui_exploiting_2021,schoppmann_make_2019} assume some public knowledge.
For example, Schoppmann et al. \cite{schoppmann_make_2019} assumes ``that the sparsity metric of the input is public.''
The definition of this ``sparsity metric'' varies in each of their algorithm: the total number of non-zeros, or the number of non-zero rows or columns.
In our paper, this public ``sparsity metric'' is the number of non-zeros per row.

Contrary to \cite{schoppmann_make_2019}, papers like \cite{yu_lodia_2025} did not identify their need for public knowledge, but it is always present.
For example, in their FHE-based sparse multiplication, Yu et al. \cite{yu_lodia_2025} split a large sparse matrix into several smaller matrices; the number and structure of the smaller matrices indirectly reveal the number of non-zeros.

Public knowledge could be interpreted as a limiting assumption, but \textbf{public knowledge is mandatory to design efficient and secure sparse algorithms}.
In sparse algorithms, the performance gains (in memory, communication, and computation) are correlated to the sparsity.
Thus, it is necessary to know the sparsity to expect any performance gain.

Anyway, \textbf{sparsity knowledge is often necessary to ML practitioners}.
ML applications can be seen as a pipeline: exploratory data analysis (to choose an appropriate ML model), ML training, and model deployment.
Data analysis on sparse dataset necessarily includes a sparsity estimation; often essential to choose an appropriate model.
Hence, the exploratory data analysis would provide this public knowledge, before the sparse matrix multiplications.

Finally, Sec. \ref{sec:minimize-prior-knowledge} proposes some \textbf{protocols to minimize the public knowledge} and Sec. \ref{sec:private-prior-knowledge} proposes protocols to obtain the minimized knowledge in a privacy-preserving manner.

\section{Secure sparse multiplications}
\label{sec:algo}
This section presents two secure sparse matrix multiplications: matrix-vector and matrix-matrix.
Existing secure sparse multiplications are incompatible with the outsourced setting, so we can only compare to the dense baselines.
We assume the bit length and number of computation parties are $O(1)$, so our asymptotic complexities are only in function of the matrix size.

Our algorithms rely on secure additions, multiplications, comparisons, oblivious shuffling, and oblivious sorting \cite{hamada_oblivious_2014}.
On the one hand, sorting-network-based oblivious sorting only requires secure additions, multiplications, and comparisons.
On the other hand, we can implement oblivious shuffling using oblivious sorting: generate a random secret-shared value for each element in the input list and obliviously sort the list based on these random values.
Our algorithms are then \textbf{compatible with all secret-sharing schemes} supporting basic arithmetic and comparison operations.
However, for optimal performances, one should implement our algorithms using honest-majority sorting/shuffling \cite{hamada_oblivious_2014,laur_round-efficient_2011}.

Appendix \ref{app:sec_proof} sketches a security proof for our algorithms.

\subsection{Toy protocol: vector-vector}
We propose to start with a simple vector multiplication.
Since dense vector multiplications have an $O(1)$ communication cost (under honest majority), we know in advance that it is unbeatable.
Modern ML applications essentially involve matrix-vector and matrix-matrix multiplications, so vector multiplication is not a relevant focus.
However, this \emph{toy protocol} provides simple intuitions to understand how to build sorting-based sparse multiplications.

We want to multiply two $n$-dimensional sparse vectors $x$ and $y$, represented as tuple lists.
Alg. \ref{alg:sparse_vect_mult} presents our solution.
First, we concatenate the two tuple lists.
Note that this concatenation is trivial (and) secure because the sparse vectors $x$ and $y$ are lists of shares.
Thus, the concatenation requires no communication; the shareholders simply concatenate locally their list of shares.
Then, we sort (\emph{obliviously}) the resulting list by coordinate.
If two consecutive tuples have equal coordinates, we multiply their values and sum all the multiplication results to build the vector multiplication value.

\begin{algorithm}[t]
	\caption{Secure sparse vector multiplication \label{alg:sparse_vect_mult}}

	\algrenewcommand\algorithmicif{$[\![\text{\textbf{obliv. if}}]\!]$ }
	\begin{algorithmic}[1]
		\Input{$\share{x}$ and $\share{y}$ are two secret-shared sparse vectors}
		\PublicKnowledge{Number of non-zero elements per vector}
		\Function{SparseVectorMult}{$\share{x},\share{y}$}
		\Let{$\share{z}$}{$\text{Concatenate}(\share{x},\share{y})$}
		\Let{$\share{z}$}{$\osort{z}$}
		\Comment{Asc. sort on the non-zero coordinate.}
		\Let{$\share{s}$}{0}
		\For{$i \gets 1 \textrm{ to } \nnz(x) + \nnz(y) -1$}
		\If{$\share{\coord(z_i)} = \share{\coord(z_{i+1})}$} \label{lst:obliv_cond}
		\Let{$\share{s}$}{$\share{s} + \share{\val(z_i)}\times\share{\val(z_{i+1})}$}
		\EndIf\label{lst:obliv_cond_end}
		\EndFor
		\State \Return{$\share{s}$} \Comment{a scalar}
		\EndFunction
	\end{algorithmic}
\end{algorithm}

We use the notation $[\![\text{\textbf{obliv. if}}]\!]$ to describe conditional structures on secret-shared values.
The presence of conditional branches may seem contrary to the obliviousness goal.
Indeed, revealing the executed branch would leak information about the secret-shared value.
However, conditional structures on secret-shared values are implemented so that the executed branch remains secret.
All the $[\![\text{\textbf{obliv. if}}]\!]$ instructions are implemented via a combination of arithmetic and comparison operations.
For example, we implement the lines \ref{lst:obliv_cond} to \ref{lst:obliv_cond_end} of Alg. \ref{alg:sparse_vect_mult} using the following expression: $\share{s}\leftarrow\share{s} + (\share{\coord(z_i)} = \share{\coord(z_{i+1})})\times\share{\val(z_i)}\times\share{\val(z_{i+1})}$; with the secure comparison $\share{\coord(z_i)} = \share{\coord(z_{i+1})}$ outputting $\share{1}$ or $\share{0}$.
Such instructions are less readable, so we keep conditional structures in our algorithms with a dedicated notation.
Such conditional statements are recurrent in the literature: e.g., \cite{asharov_efficient_2022} used conditional statements in their heavy hitters protocol.

Alg. \ref{alg:sparse_vect_mult} costs $O((\nnz(x) + \nnz(y))\log(\nnz(x) + \nnz(y)))$ communications and computations in $O(\log(\nnz(x) + \nnz(y)))$ rounds.
The sorting step induces the superlinear complexity.
The multiplication loop is parallelizable, so it only costs $O(\nnz(x) + \nnz(y))$ in $O(1)$ rounds.
In comparison, the dense multiplication has $O(1)$ communication costs under honest majority, and an $O(n\log n)$ cost under dishonest majority.
Moreover, our sparse algorithm significantly reduces the memory footprint under both threat models.

\subsection{Matrix-vector}
\label{subsec:alg-mat-vec}
We want to multiply an $n\times m$ matrix $X$ to an $m$-dimensional vector $y$.
Our work assumes the public knowledge motivated in Sec. \ref{subsec:prior-knowledge}: per-row sparsity.
Thus, we can adapt the sparse matrix representation and group the non-zero elements by rows: a sparse matrix would be a list of sparse vectors, one per row.

An intuitive solution is to compute a vector-vector multiplication on each matrix row.
This approach is adequate for dense multiplications, but inefficient for sparse matrices.
Indeed, this naive extension induces a linear dependency on the number of rows: $O((\nnz(X) + n\cdot\nnz(y))\log(\nnz(X) + n\cdot\nnz(y)))$ communications and computations.
The term $n\cdot\nnz(y)$ comes from the ``replication'' of $y$ for each row of $X$.
The vector $y$ is sorted independently for each row, which induces this inefficiency.
To sum up, we need a dedicated sparse matrix-vector multiplication because the naive extension is inefficient.

\begin{algorithm}[t]
	\caption{Secure sparse matrix-vector multiplication \label{alg:sparse_mat_vect_mult}}
	\algrenewcommand\algorithmicif{$[\![\text{\textbf{obliv. if}}]\!]$ }

	\begin{algorithmic}[1]
		\Input{$\share{X}$ a secret-shared sparse matrix, $\share{y}$ a secret-shared sparse vector}
		\PublicKnowledge{Number of non-zeros per matrix/vector}
		\Notation{the function RowCoord() (resp. ColCoord()) returns the first/row (resp. second/column) coordinate of a non-zero tuple.}
		\Function{SparseMatVectMult}{$\share{X},\share{y}$}
		\Let{$\share{z}$}{$\share{X}$}
		\For{$(\share{j}, \share{v}) \textrm{ in } \share{y}$}
		\State{Append tuple $(\coord:(\share{\bot}, \share{j}),\val:\share{v})$ to $\share{z}$}
		\EndFor

		\Let{$\share{z}$}{$\osort{z}$}
		\Comment{Asc. sort on the 2nd and then 1st coordinates (with $\bot < 1$)}
		\Statex
		\Let{$\text{prev}Y$}{$(\coord:(\share{\bot}, \share{\bot}),\val:\share{\bot})$}
		\For{$k \gets 1 \textrm{ to } \left(\nnz(X) + \nnz(y)\right)$} \label{lst:mult_loop}
		\If{$\share{\text{ColCoord}(\text{prev}Y)} = \share{\text{ColCoord}(z_k)}$}
		\Let{$\share{\val(z_k)}$}{$\share{\val(z_k)} \times \share{\val(\text{prev}Y)}$}
		\Else
		\Let{$\share{\val(z_k)}$}{$\share{0}$}
		\EndIf
		\If{$\share{\text{RowCoord}(z_k)} = \share{\bot}$}
		\Let{$\share{\text{prev}Y}$}{$\share{z_k}$}
		\Let{$\share{\val(z_k)}$}{$\share{\bot}$}
		\EndIf
		\EndFor\label{lst:mult_loop_end}
		\Statex
		\State{Drop the second coordinate from all tuples in $\share{z}$}
		\Let{$\share{z}$}{$\osort{z}$}
		\Comment{Asc. sort on the remaining coord.}
		\State{\Call{AggEqualCoord}{$\share{z}$}} \label{lst:agg_step_alg_mat_vect}
		\State{\Call{PlaceholderRemoval}{$\share{z}$}}
		\State
		\Return{$\share{z}$}
		\Comment{sparse vector stored as a list of non-zeros}
		\EndFunction
	\end{algorithmic}
\end{algorithm}

Alg. \ref{alg:sparse_mat_vect_mult} presents a more optimized solution avoiding the linear dependency on $n$.
The intuition behind this optimization is to build a tuple list sorted as follows $[y_{j}, x_{i_1,j}, x_{i_2,j}, \ldots y_{j'}, x_{i_3,j'} \ldots]$.
In this list, we group the elements related to the same column, starting with the element from vector $y$.
We then fix $y_j$, iterate over all $x_{*,j}$, and multiply it with $y_j$ until we reach the next $y_{j'}$.

\begin{algorithm}[t]
	\caption{Aggregation of tuples with equal coordinates\label{alg:naive_aggregate}}
	\algrenewcommand\algorithmicif{$[\![\text{\textbf{obliv. if}}]\!]$ }

	\begin{algorithmic}[1]
		\Input{$\share{Z}$, a sorted list of secret shared tuples}
		\Procedure{AggEqualCoord}{$\share{Z}$}
		\For{$k \gets 1 \textrm{ to } \text{length}(\share{Z}) - 1$}
		\If{$\share{\coord(Z_k)} =\share{\coord(Z_{k+1})}$}
		\Let{$\share{\coord(Z_k)}$}{$\share{\bot}$}
		\Let{$\share{\val(Z_{k+1})}$}{$\share{\val(Z_{k+1})} + \share{\val(Z_k)}$}
		\EndIf
		\EndFor
		\EndProcedure
	\end{algorithmic}
\end{algorithm}

\begin{algorithm}[t]
	\caption{Placeholder removal\label{alg:placeholder_removal}}
	\begin{algorithmic}[1]
		\Input{$\share{Z}$, a list of secret shared tuples}
		\Procedure{PlaceholderRemoval}{$\share{Z}$}
		\Let{$\share{Z}$}{$\oshuffle{Z}$}
		\For{$k \gets 1 \textrm{ to } \text{length}(\share{Z})$}
		\Let{$\text{is\_placeholder}$}{\Call{reveal}{$\share{\coord(Z_k)} = \bot$}}
		\If{$\text{is\_placeholder}$}
		\State{Remove the $k$-th tuple from $\share{Z}$}
		\EndIf
		\EndFor
		\EndProcedure
	\end{algorithmic}
\end{algorithm}

Once we finish the vector component multiplications, we sort the multiplication results according to the first coordinate to aggregate all the results with equal coordinates and build the result.
Every time two values are aggregated, the first value is replaced by a placeholder.
Alg. \ref{alg:naive_aggregate} details this aggregation step.

\begin{table}[t]
	\centering
	\resizebox{\columnwidth}{!}{
		\begin{tabular}{|c|c||c|c|c|}
			\hline
			Majority                   & Algorithm                                                                                         & Comm.                                             & Comp.       & Storage \\\hline
			\multirow{2}{*}{Dishonest} & Dense  \cite{chen_maliciously_2020,mohassel_secureml_2017,mono_implementing_2023}                 & \multicolumn{2}{|c|}{$O(nm)$}                     & $O(nm)$               \\\cline{2-5}
			                           & Sparse (\textbf{ours}: Alg. \ref{alg:sparse_mat_vect_mult})                                       & \multicolumn{2}{c|}{$O(\nnz_*\cdot\log(\nnz_*))$} & $O(\nnz_*)$           \\\hline
			\multirow{2}{*}{Honest}    & Dense  \cite{hoogh_design_2012,koti_swift_2021,lu_aegis_2023,mohassel_aby3_2018,wagh_falcon_2021} & $O(n)$                                            & $O(nm)$     & $O(nm)$ \\\cline{2-5}
			                           & Sparse (\textbf{ours}: Alg. \ref{alg:sparse_mat_vect_mult})                                       & \multicolumn{2}{c|}{$O(\nnz_*\cdot\log(\nnz_*))$} & $O(\nnz_*)$           \\\hline
		\end{tabular}}
	\caption{Complexities of the sparse and dense matrix-vector multiplications. Notations: $\nnz_*=\nnz(X) + \nnz(y), \nnz(X)\ll n\times m, \nnz(y)\ll m$.}
	\label{tab:mat_vect_mult}
\end{table}

After the aggregation (Line \ref{lst:agg_step_alg_mat_vect} of Alg. \ref{alg:sparse_mat_vect_mult}), we remove the placeholders using Alg. \ref{alg:placeholder_removal}.
To avoid revealing information about the non-zero coordinates, Alg. \ref{alg:placeholder_removal} uses the ``shuffle-and-reveal'' trick frequently used in secure sparse sum \cite{asharov_efficient_2022,bell_distributed_2022}.

Alg. \ref{alg:sparse_mat_vect_mult} costs $O((\nnz(X) + \nnz(y))\log(\nnz(X) + \nnz(y)))$ communications and computations, but it has a linear round complexity due to the multiplication loop (lines \ref{lst:mult_loop} to \ref{lst:mult_loop_end} of Alg. \ref{alg:sparse_mat_vect_mult}) and the aggregation step (Alg. \ref{alg:naive_aggregate}).
Appendix \ref{app:linear_rounds} also presents recursive algorithms taking inspiration from the logarithmic-round secure maximum (notably implemented in MPyC \cite{schoenmakers_mpycpython_2018}) to build a multiplication loop and an aggregation step with logarithmic round complexity.

Table \ref{tab:mat_vect_mult} compares the complexities of our sparse algorithm to the dense baselines.
On the one hand, honest-majority dense multiplication is no longer constant in communications, so our algorithm is advantageous for highly sparse matrices.
On the other hand, the memory footprint is more critical in this operation.
Several real-world sparse matrices do not fit in memory using a dense format (see example application in Sec. \ref{sec:exp_analysis}), even in plaintext!
\textbf{Our sparse algorithm is then mandatory when memory overflowing prevents from using dense algorithms.}

\subsection{Matrix-matrix}
\label{subsec:mat-mat}

We want to multiply two sparse matrices.
Such multiplications are used in ML algorithms for specific computations, especially correlation matrices (i.e., $X^{\top}X$).
Other ML algorithms also require sparse-dense matrix multiplication (e.g., to multiply a sparse input matrix with a dense parameter matrix).
Even in plaintext, sparse-dense multiplications require specific algorithms differing from those of sparse-sparse multiplications (i.e., the focus of our paper).
Thus, we leave sparse-dense multiplications for future works.
This subsection then emphasizes the computation $X^{\top}X$, which is a key sparse-sparse operation in ML.

\begin{algorithm}[t]
	\caption{Secure sparse matrix multiplication \label{alg:sparse_mat_mult}}

	\algrenewcommand\algorithmicif{$[\![\text{\textbf{obliv. if}}]\!]$ }
	\begin{algorithmic}[1]
		\Input{$\share{X}$ and $\share{Y}$ are two secret-shared sparse matrices}
		\PublicKnowledge{Number of non-zeros per column in $X$ (with $X^{(k)}$ a column), number of non-zeros per row in $Y$ (with $Y^{\{k\}}$ a row)}
		\Function{SparseMatMult}{$\share{X},\share{Y}$}
		\Let{$\share{Z}$}{$\emptyset$}
		\For{$k \gets 1 \textrm{ to } m$}
		\For{$i \gets 1 \textrm{ to } \text{length}(\share{X^{(k)}})$}
		\For{$j \gets 1 \textrm{ to } \text{length}(\share{Y^{\{k\}}})$}
		\Let{$\share{r_i}$}{$\share{\coord(X^{(k)}_i)}$}
		\Comment{Row coordinate}
		\Let{$\share{r_j}$}{$\share{\coord(Y^{\{k\}}_j)}$}
		\Comment{Column coordinate}
		\Let{$\share{r_v}$}{$\share{\val(X^{(k)}_i)}\times\share{\val(Y^{\{k\}}_j)}$}
		\State{Append a tuple $((\share{r_i},\share{r_j}), \share{r_v})$ to $\share{Z}$}
		\State{// Non-zero tuple: coordinates then value}
		\EndFor
		\EndFor
		\EndFor
		\Statex
		\Let{$\share{Z}$}{$\osort{Z}$}
		\Comment{Asc. sort on the coordinates.}
		\State{\Call{AggEqualCoord}{$\share{Z}$}}
		\State{\Call{PlaceholderRemoval}{$\share{Z}$}}
		\State
		\Return{$\share{Z}$}
		\Comment{sparse matrix stored as a list of non-zeros}
		\EndFunction
	\end{algorithmic}
\end{algorithm}

To compute $X^{\top}X$, we multiply an $m\times n$ sparse matrix to an $n\times m$ matrix.
With our public knowledge, we know the per-column sparsity of the first matrix (because $X$ is transposed), and per-row sparsity of the second.
Alg. \ref{alg:sparse_mat_mult} generalizes this problem: multiplying an $n\times m$ sparse matrix $X$ to an $m\times p$ sparse matrix $Y$, with per-column sparsity knowledge for $X$ and a per-row knowledge for $Y$.
We use the notation $X^{(k)}$ (resp. $Y^{\{k\}}$) to refer to the $k$-th column (resp. row) of $X$ (resp. $Y$) stored as a sparse vector.

When implementing a matrix-matrix multiplication naively, it is common to design a triple nested loop multiplying: each row $i$ of $X$ is multiplied to each column $j$ of $Y$.
The third loop is the vector multiplication on each $k$-th element of the vectors.
Linear algebra libraries usually leverage a simple optimization to speed up the dense matrix multiplication: swapping the loops on  $i$ and $k$.
This optimization significantly speeds up the implementation by reducing the memory access costs \cite{leiserson_introduction_2018}.

Our sparse algorithm reuses a similar intuition to design a fast algorithm and avoid performing a series of vector multiplications.
Our algorithms have two main steps: (1) compute all the individual scalar multiplications, (2) aggregate them using oblivious sorting.
In matrix-matrix multiplications, all values in the $k$-th column of $X$ must be multiplied by all values in the $k$-th row of $Y$.
Our public knowledge provides direct access to the $k$-th column of $X$ and the $k$-th row of $Y$, so we can compute the individual scalar multiplications by iterating on the non-zeros of $X^{(k)}$ and $Y^{\{k\}}$.
Alg. \ref{alg:sparse_mat_mult} reuses Alg. \ref{alg:naive_aggregate} and \ref{alg:placeholder_removal} for aggregation and placeholder removal.

\begin{table}[t]
	\centering
	\resizebox{\columnwidth}{!}{
		\begin{tabular}{|c|c||c|c|c|}
			\hline
			Majority                   & Algorithm                                                                           & Comm.                                                             & Comp.               & Storage           \\\hline
			\multirow{2}{*}{Honest}    & Dense   \cite{chen_maliciously_2020,mohassel_secureml_2017,mono_implementing_2023}  & $O(nm + mp + np)$                                                 & $O(nmp)$            & $O(nm + mp + np)$ \\\cline{2-5}
			                           & Sparse (\textbf{ours}: Alg. \ref{alg:sparse_mat_mult})                              & \multicolumn{2}{c|}{$O(\text{MinMult}\cdot\log(\text{MinMult}))$} & $O(\text{MinMult})$                     \\\hline
			\multirow{2}{*}{Dishonest} & Dense    \cite{hoogh_design_2012,lu_aegis_2023,mohassel_aby3_2018,wagh_falcon_2021} & $O(np)$                                                           & $O(nmp)$            & $O(nm + mp + np)$ \\\cline{2-5}
			                           & Sparse (\textbf{ours}: Alg. \ref{alg:sparse_mat_mult})                              & \multicolumn{2}{c|}{$O(\text{MinMult}\cdot\log(\text{MinMult}))$} & $O(\text{MinMult})$                     \\\hline
		\end{tabular}}
	\caption{Complexities of the sparse and dense matrix-matrix multiplications. Notations: $\text{MinMult} =\sum_{k=0}^m \nnz(X^{(k)})\cdot \nnz(Y^{\{k\}}), \nnz(X)\ll n\times m, \nnz(Y)\ll m\times p$.}
	\label{tab:mat_mult}
\end{table}

Table \ref{tab:mat_mult} compares our sparse algorithm to the dense algorithms.
The notation $\text{MinMult}=\sum_k \nnz(X^{(k)})\cdot \nnz(Y^{\{k\}})$ denotes the number of non-zero scalar multiplications required by the matrix multiplication.
It represents an intuitive lower bound: the best plaintext algorithm is $O(\text{MinMult})$.
Our algorithm \emph{only adds a logarithmic factor} compared to the plaintext sparse algorithm \cite{buluc_implementing_2011}.

\section{Experimental analysis}
\label{sec:exp_analysis}

This section compares the dense algorithms to our sparse algorithms under three sparsity levels: 99\%, 99.9\%, 99.99\% of zero values.
These thresholds are recurrent in sparse data applications (see real-world datasets mentioned in Sec. \ref{subsec:larger-sparse}).
The dense algorithm is, by definition, insensitive to the sparsity level, so we provide one curve for the dense algorithm and one per sparsity level for our sparse algorithm.

Our comparison focuses on communication costs and memory usage.
The communication costs corresponds to the total communication exchanged between the servers during the protocol.
The memory usage is represented thanks to the memory overflows: we scale up our experiments until the algorithms cannot be executed due to a lack of memory.
We use red crosses to represent these memory overflows on the plots.

For each multiplication type (i.e., matrix-vector and matrix-matrix), we also implement an example use case built upon it.
Each example use case relies on a real-world dataset to demonstrate the impracticality of dense algorithms on simple ML applications involving sparse data.

\subsection{Experimental setup}
We execute our experiments on a server with 188GB of RAM and an Intel Xeon CPU.
We use the MPyC framework \cite{schoenmakers_mpycpython_2018} to simulate 3-party protocols.

Our experiments simulate an honest majority because this setting provides the \textbf{best dense baselines} (see Tables \ref{tab:mat_vect_mult} and \ref{tab:mat_mult}).
More specifically, we chose the \emph{SSS-based matrix multiplication \cite{hoogh_design_2012} as dense baseline} instead of the other honest-majority protocols \cite{koti_swift_2021,lu_aegis_2023,mohassel_aby3_2018,wagh_falcon_2021} for two reasons.
First, the SSS-based matrix multiplication has smaller constant factors; e.g., it only requires the communication of a single ring element, while \cite{koti_swift_2021,lu_aegis_2023,mohassel_aby3_2018,wagh_falcon_2021} require multiple communication rounds.
Second, the SSS-based protocol supports an arbitrary number of parties (like our sparse multiplications), while \cite{koti_swift_2021,lu_aegis_2023,mohassel_aby3_2018,wagh_falcon_2021} are 3-party protocols.

We use 64-bit fixed-point numbers, with a 32-bit fractional part.
Floating-point arithmetic is still poorly scalable in MPC protocols.
MPC works usually rely on fixed-point arithmetic \cite{kelkar_secure_2022} to avoid performance bottlenecks.
The minor precision loss due to fixed-point arithmetic is often acceptable for ML applications \cite{kelkar_secure_2022}.

We use the protocol by Laur et al. \cite{laur_round-efficient_2011} for oblivious shuffling and the oblivious radix sort \cite{hamada_oblivious_2014} for sorting.
The first step of the oblivious radix sort is a secure bit decomposition.
As this operation is expensive, we accelerate the process using pre-computations.
Our algorithms assume the data owners pre-compute and share a bit decomposition of their non-zero coordinates.
Hence, they share a bit array instead of a single integer per non-zero coordinate.
This pre-computed bit decomposition creates a storage overhead proportional to the bit-length of the indices.
As large sparse matrices hardly scale above $10^8$ columns, so this factor would not exceed 30, and is compensated (as shown in our experiments) by the gains provided by sparse data representations.

Our implementation does not include the round-complexity optimizations of Appendix \ref{app:linear_rounds}. 

\subsection{Matrix-vector}
\begin{figure}[t]
	\centering
	\begin{subfigure}{0.49\linewidth}
		\includegraphics[width=\linewidth]{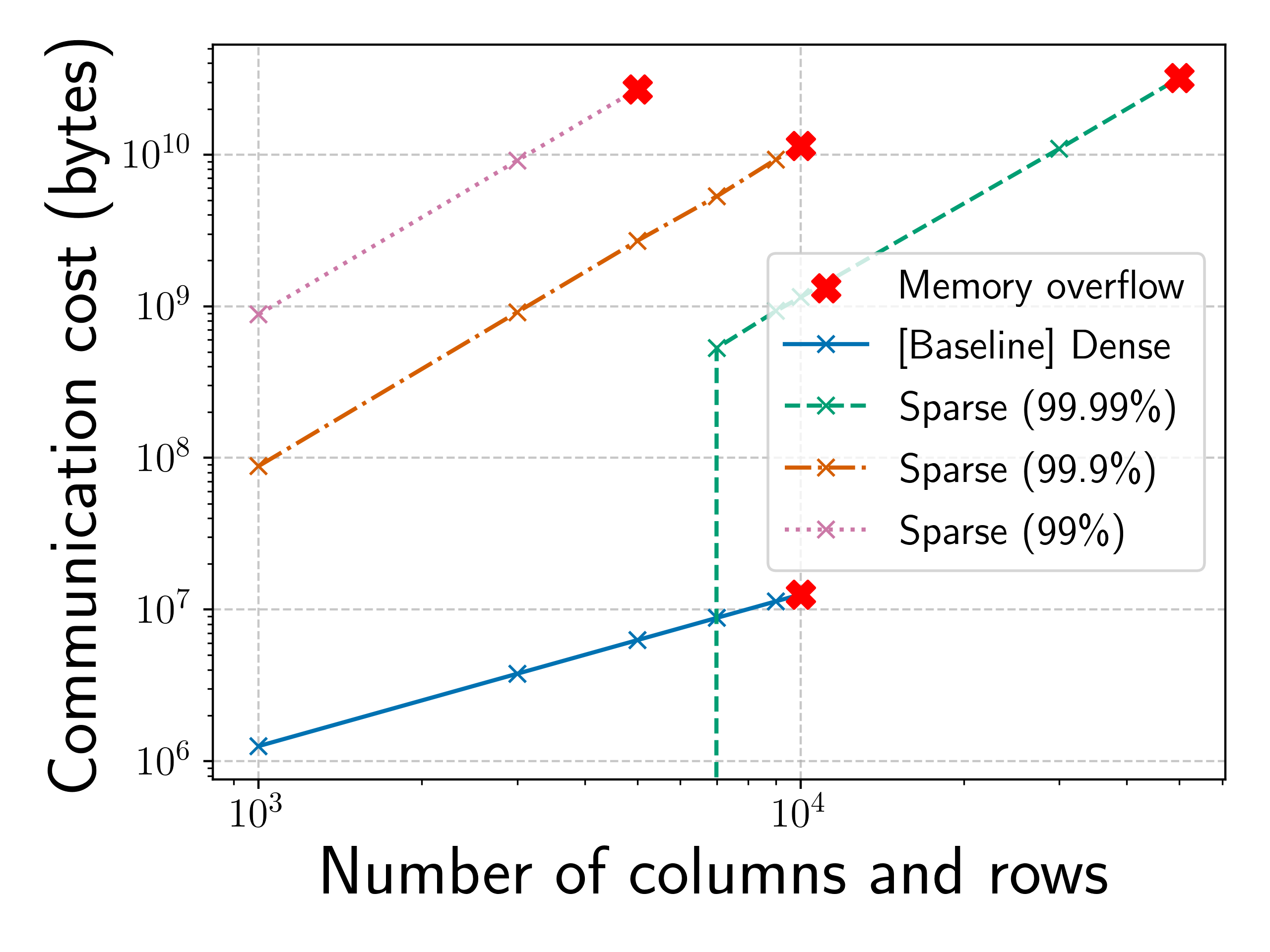}
		\Description[Plot representing the communication costs for varying square matrix size. The dense baseline outperforms the sparse algorithm in all three sparsity levels (i.e., 99\%, 99.9\% 99.99\%).]{Plot representing the communication costs for varying square matrix size from 1K to 100K. The dense baseline outperforms the sparse algorithm in all three sparsity levels (i.e., 99\%, 99.9\% 99.99\%). For 99.99\% sparsity and a matrix size of 10K, the dense algorithm is a hundred times less expensive than the sparse algorithm. The sparse algorithm cost appears to scale slightly faster than the dense algorithm costs. The dense algorithm and the sparse algorithm with 99.9\% sparsity trigger memory overflows for matrix sizes above 10K. The sparse algorithm with 99.99\% sparsity triggers a memory overflow when the size exceeds 50K.}
		\caption{Matrix-vector (with a square matrix)}
		\label{fig:benchmark_mat_vect_mult}
	\end{subfigure}
	\begin{subfigure}{0.49\linewidth}
		\includegraphics[width=\linewidth]{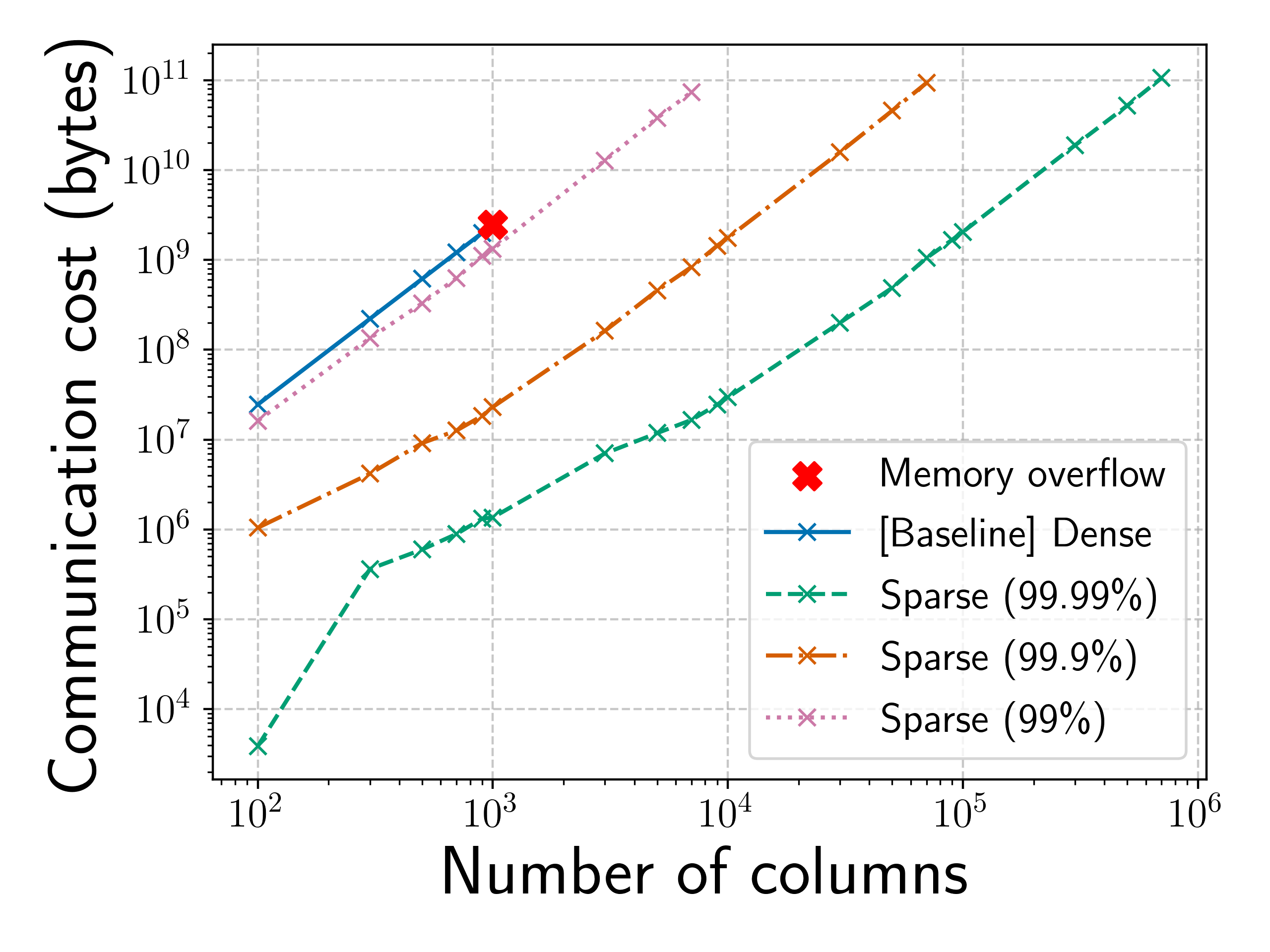}
		\Description[Plot representing the communication costs for varying numbers of columns. The dense baseline is outperformed by the sparse algorithm in all three sparsity levels (i.e., 99\%, 99.9\% 99.99\%).]{Plot representing the communication costs for varying square matrix size from 100 to 1M. The dense baseline is outperformed by the sparse algorithm in all three sparsity levels (i.e., 99\%, 99.9\% 99.99\%). For 99\% sparsity, the dense algorithm is slightly more expensive than the sparse algorithm. For 99.99\% sparsity, the dense algorithm is a thousand times more expensive than the sparse algorithm. The dense algorithm memory overflows when the number of columns exceeds 1K. The sparse algorithm with 99.99\% sparsity scales up to 1M columns without a memory overflow.}
		\caption{Matrix-matrix (operation: $X^{\top}X$).}
		\label{fig:benchmark_mat_mult}
	\end{subfigure}
	\caption{Comparison of the secure matrix multiplications under honest majority: dense algorithm vs. \emph{our sparse algorithm} with three sparsity levels.}

\end{figure}

Fig. \ref{fig:benchmark_mat_vect_mult} compares the dense and sparse matrix-vector multiplications for a square matrix.
Note that our algorithm with 99.99\% sparsity has a negligible cost when the number of rows is small because the high sparsity leads to empty sparse vectors.
Overall, the dense algorithm is less expensive than our algorithm.

\textbf{The main takeaway is the memory issue}.
Indeed, while the dense algorithm is less expensive, our algorithm supports larger matrices (i.e., the first motivation of sparse algorithms).
For example, the dense algorithm (contrary to our sparse algorithm with 99.99\% sparsity) triggers memory overflows for matrices with more than 10K columns.
As developed in Sec. \ref{subsec:larger-sparse}, there is no (dense) solution to avoid memory issues (19TB of memory would be needed).

One may argue that our algorithm also has some memory overflows.
To analyze these results, we must remind the results from Sec. \ref{subsec:larger-sparse}: larger datasets become sparser.
Understanding this phenomenon is essential to interpret Fig. \ref{fig:benchmark_mat_vect_mult}: the memory overflows illustrated in this figure should not impact most real-world datasets because of their correlation between sparsity and matrix sizes.

Our figures represent the costs for fixed sparsity rates, but large matrices are unlikely to have a low sparsity rate.
For example, real-world datasets with 10K columns \cite{li_convergence_2017} can have much more than 99\% sparsity; e.g., with 10K columns, Yahoo Movie dataset already has a 99.7\% sparsity!
Thus, using real-world sparse data, our algorithm should not be affected by such memory issues, because larger datasets are sparser.

More generally, any algorithm (even sparse) has memory issues when the data is large enough.
Sparse algorithms simply support larger sparse matrices than dense algorithms.
In practice, the correlation between sparsity and matrix size even amplifies the gains provided by sparse algorithms.

\paragraph{Example use case: Recommender system}
Matrix-vector is a core component of some recommender systems.
This example use case builds a simple recommender system in which a client sends a (secret-shared) item identifier and receives a list of similar items.
We use the \textbf{Bookcrossing dataset}: contains ratings from 279K users on 340K books, resulting in 99.998\% zeros in the dataset (1K users as training dataset).
Each row corresponds to a user, and each column to a book.

Our recommender system relies on a \textbf{nearest neighbors} algorithm \cite{aggarwal_recommender_2016}.
This model has no training step: the inference is computed on the training dataset, and requires two matrix-vector multiplications.
This algorithm computes the similarity between the book the client chose and all the other books in the database.
Then, the system returns the identifiers of the $k$ most similar books.

The dense algorithm triggers a memory overflow because it cannot store the training data using a dense matrix. 
The sparse algorithm takes, in average, 48 minutes.
\textbf{The dense algorithm is impractical while our sparse algorithms support this application}.
Our runtime could be further reduced using MPC optimizations (e.g., precomputed comparisons \cite{sun_sok_2025}).

\subsection{Matrix-matrix}
\label{subsec:exp-matrix-matrix}

Fig. \ref{fig:benchmark_mat_mult} shows the communication costs for the multiplication $X^\top X$.
This experiment considers a varying number of columns and a fixed number of rows (100 rows).
The sparse algorithm brings a significant communication reductions compared to the dense multiplication: $\times100$ gain for 99.9\% sparsity and a $\times1000$ gain for 99.99\% sparsity.
The cost reduction is much higher for matrix-matrix multiplication than for matrix-vector because the matrix-matrix multiplication complexity has a quadratic dependency on sparsity.

Besides the performance gain, \textbf{our algorithm supports much larger matrices}: the dense algorithm triggers a memory overflow for matrices with more than 1K, while our sparse algorithm scales up to 1M columns.
Fig. \ref{fig:benchmark_mat_mult} shows no memory overflow for our algorithm because we stopped experiments exceeding 10 hours.

\paragraph{Example use case: Access control}
This multiplication can be used to build an ML-based access control system.
This example use case trains an ML model deciding whether an access request is suspicious.
Access data can be particularly privacy-sensitive, e.g., the access log of medical databases can indirectly leak patient information.
As medical access logs are sensitive, there is no publicly available dataset, so we use use the \textbf{Amazon access control dataset}\footnote{\url{https://github.com/pyduan/amazonaccess/tree/master/data}}; that should have similar properties as more sensitive datasets.

The Amazon dataset contains 32.7K samples, each containing the user's properties, the service properties, and the access decision (granted or not).
Each property is a categorical value that we transform into binary arrays using the one-hot encoding.
This encoded dataset has 15K features with 99.95\% sparsity.

Our access control relies on \textbf{Linear Discriminant Analysis}.
The training requires computing: (1) the proportion of each class (i.e., access decision), (2) the mean vector of each class, and (3) the covariance matrix of the whole dataset.
The most expensive operation is, by far, the covariance matrix estimation.

The covariance matrix is necessarily dense, so we use a trick to store it efficiently.
As $Cov(X) = \frac{1}{n}X^{\top} X - \bar{x}^{\top} \bar{x}$, we store the covariance in two parts: the sparse product $X^{\top} X$ and the dense mean vector $\bar{x}$.
Covariance matrices are used at inference time via a matrix-vector multiplication (e.g., $Cov(X) \cdot y = \frac{1}{n}X^{\top} X \cdot y - \bar{x}^{\top} \bar{x}\cdot y$).
Thus, this compact representation does not change the model inference algorithm and avoids a potential memory overflow due to the high-dimensional dense matrix.

We implement this application using the dense and sparse matrix multiplications.
We use the first 10K samples as training dataset.
As expected from our previous benchmark, the dense algorithm triggers a memory overflow during the covariance computation.
The sparse algorithm completes the training algorithm in 5 hours.
Hence, \emph{our sparse algorithm provides a satisfying solution} to this problem, while the \textbf{dense algorithm is impractical}.

\section{Minimizing the public knowledge}
\label{sec:minimize-prior-knowledge}
Sec. \ref{subsec:prior-knowledge} highlighted the need for sparsity-related public knowledge in secure sparse algorithms.
While necessary, this information disclosure might be problematic in some specific use cases in which revealing the exact per-row sparsity is too sensitive.
For such cases, we propose three techniques to minimize the public knowledge necessary to our algorithms.

\subsection{Row anonymization}
\label{subsec:row-anonymization}
By default, our algorithms assume the per-row sparsity to be public.
This information means that the number of non-zeros from each data owner is public (but the position and value of the non-zeros remain secret).
To enhance their privacy, the data owner can submit their shares via an anonymization layer (e.g., Tor network).  Machine learning algorithms typically assume that training instances are independently and identically drawn and don't extract information from the order of the instances (rows).  Permuting the rows of the input matrix hence does not substantially change the input.

Thanks to the anonymization layer, the servers would be unable to link a data owner to a number of non-zeros.
This privacy enhancement also reduces the public knowledge: instead of the per-row sparsity, the distribution of the per-row sparsity is public.

This difference is subtle, but it is important for privacy: instead of individual information, a ``collective'' information is public.
Indeed, this public knowledge provides information about the whole population, but not about an individual.

ML practitioners usually need this sparsity distribution information, before performing any ML training.
Sparsity is essential information to choose an appropriate ML model.
Thus, it is realistic to assume that the distribution of the per-row sparsity is public.

\subsection{Max-row padding}
\label{subsec:row-padding}
As motivated above, the per-row sparsity distribution is acceptable in practical ML applications.
However, if needed, we can further reduce the public knowledge and only reveal an upper bound on the per-row sparsity.

To only assume this upper bound as public knowledge, the input sparse matrix must be padded to this maximum per-row sparsity.
To pad a row, the data owner represents some of the zero elements as non-zeros anyway (i.e., dummy non-zeros), such that all rows have the pre-specified number of non-zeros.
Computing the maximum per-row sparsity over all rows via MPC requires $O(n\log n)$ communications (in $O(\log n)$ rounds) for a matrix of $n$ rows.

\subsection{Matrix templating}
\label{subsec:matrix-templates}
The main drawback of max-row padding  is its sensitivity to the maximum number of non-zeros per row.
Sparse datasets usually contain only a few rows with many non-zeros.
Most rows contain much fewer non-zeros than the maximum.
For example, Fig. \ref{fig:per-row-sparsity-distrib} shows the per-row sparsity distribution on the Bookcrossing dataset.
While the maximum number of non-zeros per row is nearly $10^4$, 99\% of the rows have less than $100$ non-zeros and 90\% of the rows have less than $10$!
Thus, padding the rows until a global maximum number of non-zeros would add many unnecessary dummy non-zeros to the dataset.

Up until now, all our techniques and algorithms were oblivious to the number of rows owned by each data owner: they support all possible settings from data owners owning a single row to settings with data owners owning a sub-matrix.
Note that the input size (i.e., number of rows) of each data owner is usually public in MPC.
If each data owner owns a sub-matrix, we can propose a more subtle padding technique avoiding the shortcomings of max-row padding.

\begin{figure}
    \centering
    \includegraphics[width=.6\linewidth]{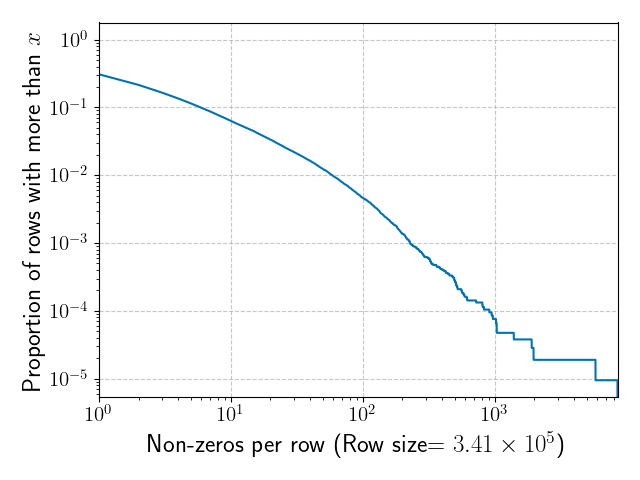}
    \caption{Per-row sparsity distribution in Bookcrossing \cite{ziegler_improving_2005}}
    \Description[Plot showing the proportion of rows having more than x non-zeros for an x varying from 1 to 10K]{Plot showing the proportion of rows having more than x non-zeros for an x varying from 1 to 10K. On a logarithmic scale, the decrease looks nearly linear.}
    \label{fig:per-row-sparsity-distrib}
\end{figure}

Instead of padding each row to a fixed maximum, we propose to build a ``matrix template'', a safe upper bound for the distribution of number of non-zeros per row, for which we are sure the real data matrix will fit in it.
Let us call  $\hat{F}(d)=\sum_{i=d}^\infty \hat{f}(i)$ the number of rows in the real dataset having $d$ non-zeros or more.
Let us describe what we call a template for a matrix $Y$: the template divides $Y$ (horizontally) into $K$ sub-matrices: $Y_1,\dots, Y_K$, with respectively $n_1, \dots, n_K$ rows.
Each row in $Y_k$ must have $\overline{\nnz_k}$ non-zero elements, and we have $\overline{\nnz_1} \le \overline{\nnz_2} \le \dots \overline{\nnz_K}$.
A matrix template is then defined by $T=\{(n_1, \overline{\nnz_1}), \dots , (n_K, \overline{\nnz_K})\}$. 
A real dataset will fit in the template (after a permutation of the rows, see Sec. \ref{subsec:row-padding}) if $\forall i \in [K]: \hat{F}(\overline{\nnz_i}+1) \le \sum_{j=i+1}^K n_j$.

To set the thresholds $\{\overline{\nnz_1},\dots,  \overline{\nnz_K}\}$, we propose to consider the quantiles: $0.25, 0.5, 0.75, 0.9, 0.99, 1$, because they capture nicely the power-law distribution of sparse data.
In other words, $\hat{F}(\overline{\nnz_1}) = 0.25, \hat{F}(\overline{\nnz_2}) = 0.5, \dots, \hat{F}(\overline{\nnz_K}) = 1$.

A data owner $P$ owning a dataset $Y$ can fit $Y$ into a template $T$ by padding $Y$ according to $T$.
First, they sort (increasingly) the rows of $Y$ based on their number of non-zeros.
Remember that the order of the rows does not matter in ML applications.
Second, they pad the first $n_1$ rows of $Y$ to have $\overline{\nnz_1}$ non-zeros.
They repeat this padding process for the next $n_2$ rows with $\overline{\nnz_2}$ non-zeros, etc. until the last $n_K$ rows with $\overline{\nnz_K}$ non-zeros.
Thanks to the sub-matrices, only $Y_K$ is padded to have the maximum number of non-zeros, the rest of the rows can have much fewer non-zeros.
Note that matrix templating with $K=1$ has the same output as max-row padding.
Similarly, the higher $K$ is, the smaller the overhead because the matrix template becomes more accurate.

We envision four privacy-preserving approaches to obtain such matrix templates.
First, one may know the population distribution, i.e., the probability $f(d)$ of drawing a row with $d$ non-zeros in the process generating the data.
Second, without such domain knowledge one may have seen a public dataset $D$ from which one could draw statistics and approximate the population distribution by inferring $\hat{F}_D(d) \ge F(d)$, before one has to process the sensitive data.
This typically leads to quite accurate results.
Third, one could start from the real data provided as input and use differential privacy to release safe upper bounds $\hat{F}_{DP}(d) \ge \hat{F}(d)$.
This strategy leading to looser bounds to protect privacy.
Fourth, we can compute it using MPC-based quantile estimation.

Sec. \ref{sec:private-prior-knowledge} details these private estimations of the matrix template.

\subsection{Overhead comparison}
Row padding and matrix templating induce some overhead due to the dummy non-zeros they add.
Fig. \ref{fig:overhead-comparison} compares the memory footprint of several alternatives: dense multiplication, no knowledge minimization, row anonymization, max-row padding, and matrix templating (using quantiles as in Sec. \ref{subsec:matrix-templates}).
The memory footprint is a simple proxy to understand the impact of knowledge minimization techniques on our algorithms'
related communication and computation costs (see also our tables in Sec. \ref{sec:algo}).
We base the comparison on four real-world datasets: SMS Spam \cite{gomez_hidalgo_content_2006}, Amazon Access Control, Bookcrossing \cite{ziegler_improving_2005}, and MovieLens \cite{harper_movielens_2016}.
For the matrix templating, our experiments simulate the MPC-based estimation with 20 data owners.

\begin{figure}
    \centering
    \includegraphics[width=.7\linewidth]{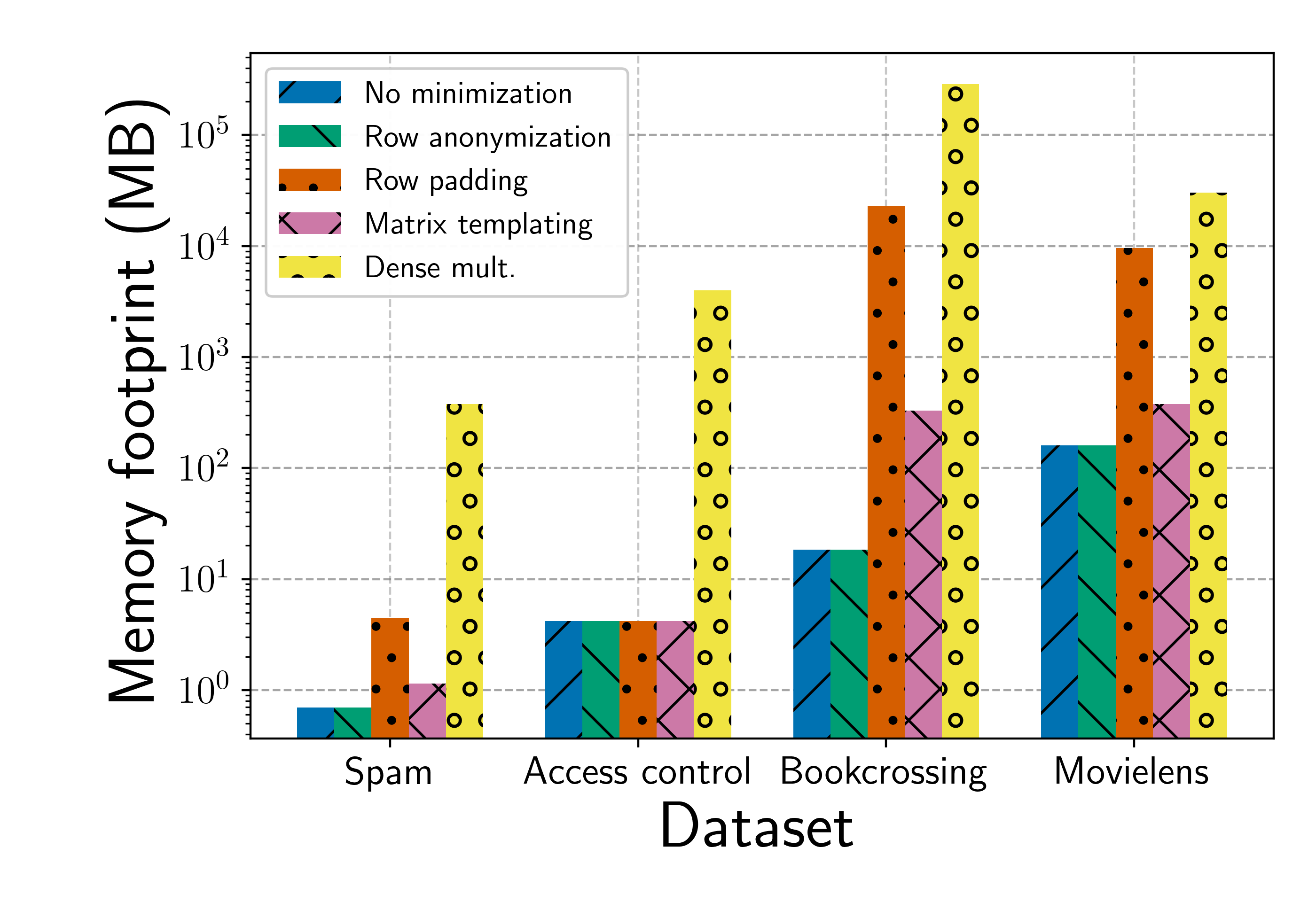}
    \caption{Matrix storage costs for different techniques of public knowledge minimization on four real-world datasets}
    \label{fig:overhead-comparison}
    \Description[This barplot presents the memory footprint of the padded matrix for 5 different settings: no knowledge mitigation, row anonymization, max-row padding, matrix templating, and dense matrix.]{This barplot presents the memory footprint of the padded matrix for 5 different settings: no knowledge mitigation, row anonymization, max-row padding, matrix templating, and dense matrix}
\end{figure}

On the one hand, (as expected) row anonymization has no effect on the memory footprint.
On the other hand, padding rows to a global maximum induces a small overhead on the Spam and the Access control datasets, while it induces significant overhead on the Bookcrossing and MovieLens datasets.
The overhead on Bookcrossing and MovieLens makes the sparse matrix nearly as large as the dense matrix... in such case, the sparse algorithms would have no interest.

Our matrix templating significantly reduces this overhead.
For example, on MovieLens, adding dummy non-zeros via matrix templating only doubles the memory footprint (vs. nearly $\times100$ for max-row padding).
As mentioned earlier, the overhead could be even further reduced if we reveal more than 5 quantiles.

Finally, on the Access Control dataset, all knowledge minimization techniques have an equal memory footprint.
In this dataset, all rows have \emph{by design} the same number of non-zeros (i.e., naturally padded), so max-row padding and matrix templating does not need to add any dummy non-zero.

\section{Private estimation of the public knowledge}
\label{sec:private-prior-knowledge}
While matrix templates minimize the necessary prior knowledge, we can further enhance privacy by estimating it using privacy-preserving techniques.
This section details how to estimate privately matrix templates based on MPC and on Differential Privacy.
Moreover, Appendix \ref{app:pop-dist} shows how to estimate them based on the population distribution.

\subsection{MPC-based}
First, each data owner $i$ shares a value $\share{\nnz_j}$ for each row $j$ they own.
Then, the servers compute securely the quantiles on the list $[\share{\nnz_1}, \dots, \share{\nnz_n}]$.
To compute them, the servers simply need to sort obliviously the list and to obtain the elements on indices: $\lfloor0.25n\rfloor,$ $\lfloor0.5n\rfloor,$ $\lfloor0.75n\rfloor,$ $\lfloor0.9n\rfloor,$ $\lfloor0.99n\rfloor,n$.

These quantiles provide an approximate template: some submatrices may not fit perfectly into it because each data owner has a slightly different row sparsity. For example, data owner $i$ may have 50\% of rows with less than 13 non-zeros while data owner $j$ has 50\% with less than 17 non-zeros. If the MPC protocol outputs 15 as quantile 0.5, data owner $i$ can pad its rows, but data owner $j$ has to remove some non-zeros to fit the template.

To finalize our template, we compute a scaling factor ensuring that all sub-matrices fit in it.
Let $\{\widehat{\nnz_1} \dots \widehat{\nnz_K}\}$ refer to the parameters from the previous approximate template.
Based on the imperfect template, each data owner $i$ computes a scaling factor $\alpha_i \ge 1$ such that: for any row $X_j$ owned by data owner $i$, $\nnz_j \le \alpha_i \widehat{\nnz_k}$ ($\widehat{\nnz_k}$ being the padding parameter corresponding to $X_j$ in the template).
Then, the data owners share their respective $\share{\alpha_i}$ and the servers compute $\overline{\alpha}$ using a secure maximum protocol on these secret-shared values.
We can now deduce the final templates: $\overline{\nnz_k} = \overline{\alpha}\widehat{\nnz_k}$, for $k\in\{1\dots K\}$.

Based the values $\{\overline{\nnz_1}, \dots, \overline{\nnz_K}\}$ and $n_i$ number of rows by data owner $i$ (usually public in MPC), we can infer the dimensions of their padded matrix: $\overline{\nnz_1}$ non-zeros in their first $\lfloor0.25n_i\rfloor$ rows, $\overline{\nnz_2}$ in their next $\lfloor(0.5 - 0.25)n_i\rfloor$ rows, ..., and  $\overline{\nnz_6}$ non-zeros in their last $\lfloor(1-0.99)n_i\rfloor$ rows.

\textbf{Using matrix templating, the public knowledge is then reduced to a few quantiles} from the per-row sparsity distribution.

\subsection{Based on Differential Privacy}
\label{subsec:dp}

\newcommand{\probDistX}{\mathcal{D}_X}
\newcommand{\epdf}{\hat{f}}
\newcommand{\ecdf}{\hat{F}}
Let $[n]$ be the set of first $n$ integers.
Let $X=\{x_i\}_{i=1}^n$ be a set of numbers identically and independently drawn from an unknown but fixed probability distribution $\probDistX$.
Let $\epdf(X,t) = |\{i\in[n]\mid x_i = t\}|$ be the empirical probability distribution (histogram) of $X$ and $\ecdf(X,t)=|\{i\in[n]\mid x_i \ge t\}|$ be the empirical cumulative distribution function.

Consider the PPML setting where each data owner has exactly one subset $X^{(j)}$, with $\mathcal{X}=\cup_{j=1}^k X^{(j)}$.  Then, the data owners can compute a secret share of $\ecdf(X,t)$ by doing computations linear in $n$ to first locally compute the partial counts $\ecdf(X^{(j)},t)$ and then summing these partial counts by MPC.

If the numbers in $X$ are sensitive, revealing $\ecdf(X,t)$ for some values of $t$ may be undesirable.  A classic approach is to use differential privacy.  In particular, two sets $X$ and $X'$ are adjacent if the sets differ in only one element.  Next, a mechanism $\mathcal{M}$ is $\epsilon$-differentially private (DP) if for all pairs of adjacent $X$ and $X'$, and for all possible subsets $Y$ of the output space of $\mathcal{M}$ there holds that
\[
	\left| \log\left(\frac{P(\mathcal{M}(X)\subseteq Y)}{P(\mathcal{P}(X')\subseteq Y)}\right) \right| \le \epsilon
\]

In contrast to classic DP mechanisms which add zero-mean noise, we are here interested in finding a DP upper bound $U_{\ecdf}$ of $\ecdf$.

\newcommand{\posLapMech}[2]{\mathcal{M}^+_{#1,#2}}
\newcommand{\posLapMechED}{\posLapMech{\epsilon}{\delta}}
Let us first consider finding a DP upper bound to a single number $y=\ecdf(X,t)$.  We define $\posLapMechED(y) = y - \frac{1}{\epsilon}\log(2\delta) + \hbox{Lap}\left(\frac{1}{\epsilon}\right)$, where $\hbox{Lap}(b)(y)=\exp(-|y|/b)/2b$ is the Laplace distribution.

\begin{lemma}
	\label{lm:single.dp.upper}
	$\posLapMechED$ is $\epsilon$-DP and
	with probability $1-\delta$ there holds $y \le \posLapMechED(y)$.
\end{lemma}
\begin{proof}
	For $\lambda>0$, the probability that $Lap(1/\epsilon)\le -\lambda$ is
	\[\int_{-\infty}^{-\lambda} \exp(-t\epsilon)\epsilon/2 \,\hbox{d}t = \frac{1}{2} \exp(-\lambda\epsilon) \]
	Set $\lambda = -\log(2\delta)/\epsilon$, then the probability that $Lap(1/\epsilon)\le -\lambda$ is $\delta$.  So, the probability that $\posLapMechED(y) = y+ \lambda + \hbox{Lap}(1/\epsilon) < y$ is $\delta$.
\end{proof}

Next, let us consider the problem of finding a DP vector $(u_i)_{i=1}^l$ such that $\forall i\in [l] : \ecdf(X,t_i)\le u_i$ with high probability.  A well-known result on differential privacy describing the relation between the privacy cost $\epsilon$ and the needed amount of noise for histograms and similar data structures was first introduced by Dwork \cite{dwork_differential_2010}.  The main idea is that for two adjacent datasets $X$ and $X'$, the sequence $(\ecdf(X,t_i)-\ecdf(X',t_i))_{i=1}^l$ starts with 0 or more zeros, followed by a series of 0 or more $+1$ or $-1$ values, followed by 0 or more zeros.  This is a consequence of the assumption that $X$ and $X'$ differ in only one element, say $x_T$ and $x_T^\prime$, and hence this causes the sequence to only be different at the thresholds $t_i$ between the smallest and the largest of $x_T$ and $x_T^\prime$.  As a consequence, the series $(\ecdf(X,t_i))_{i=1}^l$ and $(\ecdf(X',t_i))_{i=1}^l$ correlate strongly and a strategy adding correlated noise gives favorable results (low variance and low $\epsilon$).  In particular, let us assume $l=2^L$ is a power of $2$ (else, we can just add a few thresholds $t_i$ and benefit from determining DP upper bounds to $\ecdf(X,t_i)$ for these thresholds). For $i\in[l]$ and $j+1\in [L]$, let $b_{i,j}$ give the bit decomposition of $i-1$ in the sense that $i-1=\sum_{j=0}^{L-1} 2^j b_{i,j}$.
\newcommand{\posHistMech}[2]{\mathcal{M}^{+l}_{#1,#2}}
\newcommand{\posHistMechED}{\posHistMech{\epsilon}{\delta}}
Let $T=(t_i)_{i=1}^l$ be a set of thresholds.
For $i\in[L]$ and $j+1\in[2^{i-1}]$, let $\eta_{i,j}=Lap((L+1)/\epsilon)$ be Laplace random variable.
Then, for $j+1\in[L]$ and $k \in [2^j]$ we define
\[
	\posHistMechED(X,t_i)=\ecdf(X,t_i) + \frac{L(L+1)}{\epsilon}\log\left(\frac{L(L+1)}{2\delta}\right)
	+ \sum_{j=0}^{L-1} \eta_{j+1,\lfloor (i-1)/2^j\rfloor}
\]

\begin{lemma}
	Let $\delta = \frac{\lambda^L e^{-\lambda}}{L!}$
	The complete sequence $\left(\posHistMechED(X,t_i)\right)_{i=1}^l$ is $\epsilon$-DP and with probability at least $1-\delta$ there holds $\forall i\in[l]: \ecdf(X,t_i) \le \posHistMechED(X,t_i)$.
\end{lemma}
\begin{proof}
	That the sequence  $\left(\posHistMechED(X,t_i)\right)_{i=1}^l$ is $\epsilon$-DP follows directly from a result by Dwork \cite[thm 4.1]{dwork_differential_2010}.
	Similarly to Lemma \ref{lm:single.dp.upper}.
	For the second part of the statement, we need to analyze the probability that $
		\sum_{j=0}^{L-1} \eta_{j+1,\lfloor (i-1)/2^j\rfloor}\ge -\frac{L(L+1)}{\epsilon}\log\left(\frac{L(L+1)}{2\delta}\right)$.

	The lefthand side is a sum of $L$ independent Laplace distributions $Lap((L+1)/\epsilon)$, we then have $\Pr\left(Lap\left(\frac{L+1}{\epsilon}\right)\le - \frac{L+1}{\epsilon}\log\left(\frac{L(L+1)}{2\delta}\right)\right) \le \frac{\delta}{L}$.
	So, the probability that this equality holds for all $L$ terms is $\delta$.
\end{proof}

\begin{figure}
	\centering
	\includegraphics[width=.5\linewidth]{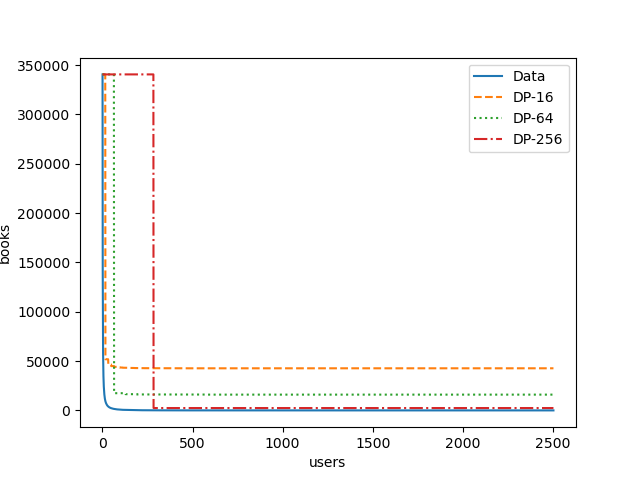}
	\caption{Distribution of row non-zero counts, and DP safe upper bounds for variable template block sizes}
	\Description[Plot showing the distribution of row non-zero counts, and DP safe upper bounds for variable template block sizes]{Plot showing the distribution of row non-zero counts, and DP safe upper bounds for variable template block sizes}
	\label{fig:dp.recomm}
\end{figure}

Fig. \ref{fig:dp.recomm} plots for the Bookcrossing dataset the real cumulative distribution of non-zeros per row; gives how many books were rated by $x$ or more users (for any $x$).
As can be seen, a few books were rated by a large number of users and most books were only rated by a few users.
The figure also shows several curves where DP statistics of the empirical cumulative distribution function were used to publish upper bounds for the number of non-zero slots needed per row.
The several curves use a different number of rows per block in the template matrix (the quantiles discussed in Sec. \ref{sec:minimize-prior-knowledge}).
In reality it is likely for most datasets (and also this one) we could perform better by taking smaller blocks for small $x$ and larger blocks for larger $x$.

We used $0.01$-differential privacy in Fig. \ref{fig:dp.recomm}, and the upper bounds are guaranteed to give sufficient space in the template for the data with probability at least $1-10^6$.  In case the data owners would in the end send rows not fitting in the template, the servers would need to abort the protocol and restart.  In practice, machine learning algorithms usually use values for $\epsilon$ larger than $0.01$.

Revealing in a DP way more quantiles requires adding more noise, which in turns results in a higher upper bound.
Too many template partitions (and hence quantiles) will unnecessarily closely model the data and suffer from the differential privacy protection, while too few quantiles will insufficiently take into account the decrease of the function.
In contrast to Fig. \ref{fig:per-row-sparsity-distrib} for the same dataset, Fig. \ref{fig:dp.recomm} uses a linear scale, making the memory footprint more visual.  In particular, the memory footprint is the area under the curve of a certain strategy.
It seems that in this case blocks of 64 rows is the best among the three options shown.

\section{Conclusion}
Our paper introduced algorithms based on oblivious sorting to multiply secret-shared sparse matrices.
Our algorithms are \textbf{compatible with a more generic setting} (i.e., outsourced MPC) than existing MPC works on sparse multiplications.
Our algorithms \textbf{avoid memory issues} present in dense secure multiplications by leveraging the data sparsity.
Our experiments show \textbf{communication cost reductions} up to $\times 1000$ compared to dense matrix multiplications.
We also implemented two real-world ML applications highlighting that \textbf{our sparse algorithms support applications impractical with existing protocols}.
Finally, we proposed methods to \textbf{minimize the public knowledge} (mandatory in any secure sparse algorithm) and obtain it in a privacy-preserving manner.

Our implementation is publicly available and open-source, so our algorithms can be transferred into MPC frameworks:
\url{https://github.com/MarcT0K/Secure-sparse-multiplications}

\begin{acks}
	This work was supported by the Netherlands Organization for Scientific Research (NWO) in the context of the SHARE project [CS.011], the ANR project ANR-20-CE23-0013 PMR, the ANR REDEEM project (ANR-23-PEIA-005), and the Horizon Europe project HE TRUMPET.
\end{acks}

\bibliographystyle{ACM-Reference-Format}
\bibliography{ref}

\appendix
\section{Security proof sketch}
\label{app:sec_proof}
This appendix sketches the security proof of our algorithms.
We provide proof sketches in the honest-but-curious model using the real-ideal paradigm \cite{evans_pragmatic_2018}.

We can define the ideal functionality for the sparse vector multiplication $\mathcal{F}_\mathsf{VectMult}^{\eta}$ as follows:
\begin{itemize}
	\item Input:  $\share{x_a}= \{\left(\share{i_{1,a}}, \share{v_{1,a}}\right),\ldots \left(\share{i_{\eta,a}}, \share{v_{\eta,a}}\right)\}$ and $\share{x_b}$ (defined similarly)
	\item Output: a secret-shared scalar $\share{s}$ s.t. $s=\sum_{j,k\in\{1\dots \eta\}} v_{j,a}\cdot v_{k,b} \cdot \delta(i_{j,a}, i_{k,b})$ (with $\delta$ the Kroenecker delta).
\end{itemize}

This functionality has a property $\eta$.
This property represents the protocol public knowledge in our functionality.
As motivated in Sec. \ref{subsec:prior-knowledge}, public knowledge is necessary to have efficient secure sparse multiplications.

We sketch a proof showing that Alg. \ref{alg:sparse_vect_mult} securely realizes the ideal functionality $\mathcal{F}_\mathsf{VectMult}^{\eta}$.
Our proof relies essentially on the composition theorem for the honest-but-curious model (Theorem 7.3.3 of \cite{goldreich_foundations_2009}).
Certain subprotocols we use have been proven to be secure and correct in the honest-but-curious model.
By using the composition theorem, we can rely on their security and only prove the security of our overall protocol in which these subprotocols are embedded.

The correctness of Alg. \ref{alg:sparse_vect_mult} can be verified by an easy arithmetic exercise.

\textit{Security of the sub-protocols}
The first statement in Alg. \ref{alg:sparse_vect_mult} is a concatenation that is performed locally with no communication and hence does not need to be simulated.
The second statement is an oblivious sort.
The proof depends on the sub-protocol chosen for instantiation.
This choice is highly influenced by the threat model, so we consider the Batcher's sort for dishonest majority and the oblivious radix sort for honest majority (i.e., the two threat models studied in our paper).
On the one hand, the Batcher's sort only requires comparisons and arithmetic operations.
It is then secure because the arithmetic and comparison operations on secret-shared values are trivially secure in the honest-but-curious model \cite{hoogh_design_2012}.
On the other hand, the oblivious radix sort is proven secure in \cite{hamada_oblivious_2014}.
After the sorting statement, we have a loop.
This loop statement reveals no information because it uses $\eta$, a property of $\mathcal{F}_\mathsf{VectMult}^{\eta}$ (i.e., a public value).
Finally, the loop contains a conditional statement.
As highlighted in Sec. \ref{sec:algo}, we use the notation $[\![\text{\textbf{obliv. if}}]\!]$ to maximize the readability.
This conditional statement would be implemented as follows: $\share{s}\leftarrow\share{s} + (\share{\coord(z_i)} = \share{\coord(z_{i+1})})\times\share{\val(z_i)}\times\share{\val(z_{i+1})}$, with $\share{\coord(z_i)} = \share{\coord(z_{i+1})}$ outputting $\share{1}$ or $\share{0}$.
The arithmetic operations are trivially secure in the honest-but-curious model.
For the secure equality, we can use the secure protocol proposed by Toft \cite{toft_primitives_2007}.
Hence, this oblivious conditional statement is secure because it relies on secure comparisons and arithmetic operations.
Such oblivious conditional statements are recurrent in MPC protocols (e.g., in private heavy hitters \cite{asharov_efficient_2022}).

\paragraph{Security of the overall protocol}
Since all the sub-protocols are individually secure, according to the Composition Theorem, Alg. \ref{alg:sparse_vect_mult} preserves the input secrecy in the honest-but-curious model.

\paragraph{Matrix-vector and matrix-matrix}
We do not develop the proof sketch for our sparse matrix-vector (Alg. \ref{alg:sparse_mat_vect_mult}) and matrix-matrix (Alg. \ref{alg:sparse_mat_mult}) multiplications, but the proof is very similar.
As our sparse vector multiplication, these two algorithms use oblivious sorting, comparisons, arithmetic operations, and oblivious conditional structures.
For all these operations, the security would be proven as sketched for the vector multiplication.
However, they also use two additional sub-protocols: AggEqualCoord and PlaceholderRemoval.

The AggEqualCoord relies on an oblivious conditional structure and arithmetic operations.
The proof is then trivial, because we can express the conditional structure using a combination of comparisons and arithmetic operators.

The PlaceholderRemoval uses a ``shuffle-and-reveal'' trick: shuffle a list of secret-shared values and reveal which of them are placeholders.
Other security papers \cite{hamada_oblivious_2014} used similar ``shuffle-and-reveal'' tricks in protocols with proven security.
In our algorithms, revealing the number of placeholders indirectly reveals the number of non-zero elements in the output matrix.
However, this information is public by definition because it is contained in the output.

\paragraph{Extension to malicious security}
Our security proofs hold in the honest-but-curious model.
Malicious security would require a verifiable secret-sharing scheme and dedicated proofs for the security in the malicious model.
We remind that our main sub-protocols have known maliciously secure variant (especially oblivious sorting \cite{asharov_efficient_2022,hamada_oblivious_2014} and oblivious shuffling \cite{asharov_efficient_2022,laur_round-efficient_2011}) that would be the basis of a malicious security proof.
We leave this extension for future work.

\section{Avoiding the linear round complexity}
\label{app:linear_rounds}

This appendix presents optimized algorithms to avoid the linear round complexity in Alg. \ref{alg:sparse_mat_vect_mult} and \ref{alg:sparse_mat_mult}.
In Alg. \ref{alg:sparse_mat_vect_mult}, this round complexity has two sources: the aggregation step (Alg. \ref{alg:naive_aggregate}) and the multiplication loop: lines \ref{lst:mult_loop} to \ref{lst:mult_loop_end}.
In Alg. \ref{alg:sparse_mat_mult}, the round complexity is only caused by the aggregation step.

\begin{algorithm}
	\caption{Secure maximum via a recursive function \label{alg:sec_max}}

	\begin{algorithmic}[1]
		\Input{$\share{V}$, an unsorted list of values.}
		\Function{RecursiveMax}{$\share{V}$}
		\Let{$n$}{$\text{length(\share{V})}$}
		\If{$n=1$}
		\State \Return{$\share{V_1}$}
		\EndIf
		\Statex

		\algrenewcommand\algorithmicif{$[\![\text{\textbf{obliv. if}}]\!]$ }

		\State{Initialize $\share{\text{Roots}}$, a list of $\lceil n/2\rceil$ secret-shared values}
		\For{$k \gets 1 \textrm{ to } \lceil n/2\rceil$}
		\If{$\share{V_{2k}} >  \share{V_{2k+1}}$}
		\Let{$\share{\text{Roots}_k}$ }{$\share{V_{2k}}$}
		\Else
		\Let{$\share{\text{Roots}_k}$ }{$\share{V_{2k+1}}$}
		\EndIf
		\EndFor

		\State \Return \Call{RecursiveMax}{$\text{Roots}$}
		\EndFunction
	\end{algorithmic}
\end{algorithm}

To avoid this round complexity, we take inspiration from the recursive secure maximum implemented in many MPC libraries (including MPyC).
If implemented naively (i.e., scanning the list elements one by one), the round complexity of the maximum function would be linear.
MPC libraries rely on a recursive algorithm to reach a logarithmic complexity.
Alg. \ref{alg:sec_max} details this algorithm.
The intuition behind it is to represent the list using a tree.
Then, each sub-tree recursively computes its maximum and ``propagate'' this value toward the root.

We reuse the same recursive structure to build our optimized algorithms with logarithmic round complexity.
Alg. \ref{alg:optim_aggregate} presents the optimized aggregation algorithm.
Alg. \ref{alg:optim_mult_loop} presents the optimized multiplication loop.
Due to their recursive structure, these algorithms are more convoluted than their naive alternatives, but these optimizations are necessary to respect the round complexities presented in the main text for the sparse matrix-vector and matrix-matrix multiplications.
These two algorithms are similar and reuse the same intuitions; only a few minor operations differ.
Hence, this appendix only discusses Alg. \ref{alg:optim_aggregate} in detail.

\begin{algorithm}[t]
	\caption{Optimized aggregation of tuples with equal coordinates\label{alg:optim_aggregate}}

	\algrenewcommand\algorithmicif{$[\![\text{\textbf{obliv. if}}]\!]$ }
	\begin{algorithmic}[1]
		\Input{$\share{Z}$, a sorted list of $n$ secret-shared tuples.}
		\PublicKnowledge{None}

		\Statex
		\Function{OptimizedAggEqualCoord}{$\share{Z}$}
		\State{// Pre-processing: $O(1)$ rounds}
		\State{Initialize list $\share{\text{Children}}$ with $\lceil n/4\rceil$ empty elements}
		\For{$k \gets 1 \textrm{ to } \lceil n/4\rceil$}
		\State{\Call{AggIfEqual}{$Z_{4k},Z_{4k+1}$}}
		\State{\Call{AggIfEqual}{$Z_{4k+1},Z_{4k+2}$}}
		\State{\Call{AggIfEqual}{$Z_{4k+2},Z_{4k+3}$}}
		\Let{$\share{\text{Children}_k}$}{$(\share{Z_{4k}}, \share{Z_{4k+1}}, \share{Z_{4k+2}}, \share{Z_{4k+3}})$}
		\EndFor

		\Statex
		\State{// Online phase: $O(\log n)$ rounds}
		\Let{$\share{\text{Z}}$}{\Call{RecProp}{$\share{\text{Children}}$}}
		\Statex
		\State{// Offline post-processing}
		\For{$k \gets 1 \textrm{ to } \lceil n/4\rceil$}
		\Let{$(\share{Z_{4k}}, \share{Z_{4k+1}}, \share{Z_{4k+2}}, \share{Z_{4k+3}})$}{$\share{\text{Children}_k}$}
		\EndFor

		\For{$k \gets 1 \textrm{ to } n$}
		\If{$\share{\val(Z_k)} \neq \share{0}$}
		\Let{$\share{\val(Z_k)}$}{$\share{\bot}$}
		\EndIf
		\EndFor

		\State \Return $\share{\text{Z}}$
		\EndFunction
	\end{algorithmic}
\end{algorithm}

\begin{algorithm}[t]
	\caption{Sub-functions used in Alg. \ref{alg:optim_aggregate}}

	\algrenewcommand\algorithmicif{$[\![\text{\textbf{obliv. if}}]\!]$ }
	\label{subalg:optim_aggregate}
	\begin{algorithmic}[1]
		\Procedure{AggIfEqual}{$\share{\text{tup}_1}, \share{\text{tup}_2}$}
		\If{$\share{\coord(\text{tup}_1)} =\share{\coord(\text{tup}_2)}$}
		\Let{$\share{\val(\text{tup}_2)}$}{$\share{\val(\text{tup}_2)} + \share{\val(\text{tup}_1)}$}
		\Let{$\share{\val(\text{tup}_1)}$}{$\share{0}$}
		\EndIf
		\EndProcedure

		\Statex
		\Function{UpProp}{$\share{\text{Child}_1}, \share{\text{Child}_2}$}
		\Let{$\share{\text{Root}}$}{$(\minleft(\share{\text{Child}_1}),\maxright(\share{\text{Child}_1}),$ $\minleft(\share{\text{Child}_2}),\maxright(\share{\text{Child}_2}))$}
		\State{\Call{AggIfEqual}{$\minleft(\share{\text{Root}}),\maxleft(\share{\text{Root}})}$}
		\State{\Call{AggIfEqual}{$\maxleft(\share{\text{Root}}),\minright(\share{\text{Root}})}$}
		\State{\Call{AggIfEqual}{$\minright(\share{\text{Root}}),\maxright(\share{\text{Root}})}$}
		\State \Return $\share{\text{Root}}$

		\EndFunction
		\Statex
		\Procedure{DownProp}{$\share{\text{Child}},$ $\share{\text{NewMin}},$ $\share{\text{NewMax}}$}
		\Let{$\minleft(\share{\text{Child}})$}{$\text{\share{\text{NewMin}}}$}
		\State{\Call{AggIfEqual}{$\minleft(\share{\text{Root}}),\maxleft(\share{\text{Root}})}$}
		\State{\Call{AggIfEqual}{$\maxleft(\share{\text{Root}}),\minright(\share{\text{Root}})}$}
		\Let{$\maxright(\share{\text{Child}})$}{$\share{\text{NewMax}}$}
		\EndProcedure
	\end{algorithmic}
\end{algorithm}

Our algorithm builds a binary tree from the non-zero tuples.
Each leaf contains four non-zero tuples, and the internal nodes represent some sub-lists of non-zero tuples.
Each internal node stores four non-zero tuples: ``minimum tuple'' of the left (child) sub-tree, ``maximum tuple'' of the left sub-tree, ``minimum tuple'' of the right sub-tree, and ``maximum tuple'' of the right sub-tree.
We compare the tuples based on their coordinate, so the maximum tuple is the tuple with the highest coordinates.

To understand the intuition behind Alg. \ref{alg:optim_aggregate}, it is necessary to understand the role of the recursion.
Let us assume we have two sub-lists on which the aggregation is already completed, and let us understand the necessary steps to obtain a whole list with the aggregation completed.
To know whether a sub-list must be updated, we only need to compare the ``maximum'' tuple of the left sub-list to the ``minimum'' tuple of the right sub-list.
This claim holds because our aggregation algorithm takes as \emph{input a sorted list} of non-zero tuples.
If these two tuples share the same coordinate, the ``maximum'' tuple of the left list must be aggregated to the ``minimum'' tuple in the right list.
Otherwise, no operation is necessary.

\begin{algorithm}[t]
	\caption{Generic recursive propagation used in Alg. \ref{alg:optim_aggregate} and \ref{alg:optim_mult_loop} \label{alg:rec_prop}}
	\begin{algorithmic}[1]
		\Input{$\share{\text{Children}}$, a list of 4-element tuples. UpProp and DownProp implementations vary depending on the algorithm using this recursive approach.}
		\Function{RecProp}{$\share{\text{Children}}$}
		\Let{$n$}{$\text{length(\share{\text{Children}})}$}
		\If{$n=1$}
		\State \Return{$\share{\text{Children}}$}
		\EndIf

		\Statex
		\For{$k \gets 1 \textrm{ to } \lceil n/2\rceil$}
		\Let{$\share{\text{Roots}_k}$ }{\Call{UpProp}{$\share{\text{Children}_{2k}},\share{\text{Children}_{2k+1}}$}}
		\EndFor
		\Let{$\share{\text{Roots}}$}{\Call{RecProp}{$\share{\text{Roots}}$}}
		\For{$k \gets 1 \textrm{ to } \lceil n/2\rceil$}
		\State{\Call{DownProp}{$\share{\text{Children}_{2k}},$ $\minleft(\share{\text{Roots}_k}),$ $\maxleft(\share{\text{Roots}_k})$}}
		\State{\Call{DownProp}{$\share{\text{Children}_{2k+1}},$ $ \minright(\share{\text{Roots}_k}),$ $\maxright(\share{\text{Roots}_k})$}}
		\EndFor
		\State \Return $\share{\text{Children}}$
		\EndFunction
	\end{algorithmic}
\end{algorithm}

\begin{algorithm}[t]
	\caption{Optimized multiplication loop for Alg. \ref{alg:sparse_mat_vect_mult}}
	\label{alg:optim_mult_loop}

	\algrenewcommand\algorithmicif{$[\![\text{\textbf{obliv. if}}]\!]$ }
	\begin{algorithmic}[1]
		\Input{$\share{Z}$, a sorted list of $n$ secret-shared tuples.}
		\PublicKnowledge{None}

		\Statex
		\Function{OptimizedMultLoop}{$\share{Z}$}
		\State{// Pre-processing $O(1)$ rounds}
		\Let{$V_1$}{$\share{\val(Z_1)}$}
		\For{$k \gets 2 \textrm{ to } n$}
		\If{$\share{\coord(Z_k)} = \share{\coord(Z_{k-1})}$}
		\Let{$V_k$}{$\share{\val(Z_k)}$}
		\Else
		\Let{$V_k$}{$\share{\bot}$}
		\EndIf
		\EndFor

		\State{Initialize list $\share{\text{Children}}$ with $\lceil n/4\rceil$ empty elements}
		\For{$k \gets 1 \textrm{ to } \lceil n/4\rceil$}
		\State{\Call{ReplaceIfNull}{$V_{4k},V_{4k+1}$}}
		\State{\Call{ReplaceIfNull}{$V_{4k+1},V_{4k+2}$}}
		\State{\Call{ReplaceIfNull}{$V_{4k+2},V_{4k+3}$}}
		\Let{$\share{\text{Children}_k}$}{$(\share{V_{4k}}, \share{V_{4k+1}}, \share{V_{4k+2}}, \share{V_{4k+3}})$}
		\EndFor

		\Statex
		\State{// Online phase: $O(\log n)$ rounds}
		\Let{$\share{\text{Children}}$}{\Call{RecProp}{$\share{\text{Children}}$}}
		\State{// Post-processing: $O(1)$ rounds}
		\For{$k \gets 1 \textrm{ to } \lceil n/4\rceil$}
		\Let{$(\share{V_{4k}}, \share{V_{4k+1}}, \share{V_{4k+2}}, \share{V_{4k+3}})$}{$\share{\text{Children}_k}$}
		\EndFor

		\For{$k \gets 1 \textrm{ to } n$}
		\If{$\share{V_k} \neq \share{\bot}$}
		\Let{$\share{\val(Z_k)}$}{$\share{\val(Z_k)}\times\share{V_k}$}
		\EndIf
		\EndFor
		\State \Return $\share{\text{Z}}$
		\EndFunction

		\Statex{}
		\Procedure{ReplaceIfNull}{$\share{v_\text{old}}, \share{v_\text{new}}$}
		\If{$\share{v_\text{old}} =\share{\bot}$}
		\Let{$\share{v_\text{old}}$}{$\share{v_\text{new}}$}
		\EndIf
		\EndProcedure
		\Function{UpProp}{$\share{\text{Child}_1}, \share{\text{Child}_2}$}
		\Let{$\share{\text{Root}}$}{$(\minleft(\share{\text{Child}_1}),\maxright(\share{\text{Child}_1}),$ $\minleft(\share{\text{Child}_2}),\maxright(\share{\text{Child}_2}))$}
		\State{\Call{ReplaceIfNull}{$\minleft(\share{\text{Root}}),\maxleft(\share{\text{Root}})}$}
		\State{\Call{ReplaceIfNull}{$\maxleft(\share{\text{Root}}),\minright(\share{\text{Root}})}$}
		\State{\Call{ReplaceIfNull}{$\minright(\share{\text{Root}}),\maxright(\share{\text{Root}})}$}
		\State \Return $\share{\text{Root}}$
		\EndFunction
		\Procedure{DownProp}{$\share{\text{Child}},$ $\share{\text{NewMin}},$ $\share{\text{NewMax}}$}
		\State{\Call{ReplaceIfNull}{$\minleft(\share{\text{Child}}),\text{\share{\text{NewMin}}}}$}
		\State{\Call{ReplaceIfNull}{$\maxleft(\share{\text{Child}}),\text{\share{\text{NewMin}}}}$}
		\State{\Call{ReplaceIfNull}{$\minright(\share{\text{Child}}),\text{\share{\text{NewMin}}}}$}
		\Let{$\maxright(\share{\text{Child}})$}{$\share{\text{NewMax}}$}
		\EndProcedure
	\end{algorithmic}
\end{algorithm}

Contrary to the secure maximum algorithm, Alg. \ref{alg:optim_aggregate} passes through the tree twice: from leaves to root, then from root to leaves.
The first phase identifies which consecutive sub-trees share minimum-maximum with equal coordinates.
The second phase propagates the aggregated values to the leaves.
Alg. \ref{alg:rec_prop} describes this recursive propagation.
The implementation of the ``upward'' and ``downward'' propagations of Alg. \ref{alg:optim_aggregate} are given in Alg. \ref{subalg:optim_aggregate}.

As each node stores four non-zero tuples, Alg. \ref{alg:optim_aggregate} implicitly assumes that the input list is a multiple of 4 in length.
If not, we can pad the list with placeholders without impacting the rest of the algorithm.

Since we rely on a binary-tree approach, the complexity of this algorithm is $O(m \log m)$ with $m$ the number of leaves.
This algorithm reduces the round complexity to $O(\log m)$ because we can parallelize the operations at each tree level.

The optimized multiplication loop (Alg. \ref{alg:optim_mult_loop}) reuses the same intuition as the optimized aggregation.
The main difference is that we do not want to aggregate some values but replicate some values.
We want to build a vector containing one multiplicative value for each non-zero tuple.
The algorithm starts by creating a vector of placeholders with a few multiplicative values.
Then, we use our recursive propagation to replace each placeholder with the closest (non-placeholder) value on the left.
Once this vector is built, we multiply our non-zero elements by their corresponding multiplicative value.
This element-wise vector multiplication requires one communication round.

Contrary to Alg. \ref{alg:optim_aggregate}, the recursive part of Alg. \ref{alg:optim_mult_loop} work on a vector of scalar instead of working on a vector of tuples.
This difference slightly simplifies the value propagation but requires more pre- and post-processing than Alg. \ref{alg:optim_aggregate}.
Hence, Alg. \ref{alg:optim_mult_loop} also uses the recursive Alg. \ref{alg:rec_prop} but with different ``upward'' and ``downward'' propagations described in Alg. \ref{alg:optim_mult_loop}.
As Alg. \ref{alg:optim_aggregate}, this optimized multiplication loop has a logarithmic round complexity.

\section{Matrix template estimation based on the population distribution}
\label{app:pop-dist}
This appendix considers the situation where either the population distribution of the non-zero counts of rows is known publicly, or one has a public sample from which one can compute statistics.

Assume that we know the population distribution $f$ where $f(d)$ is the probability that a row has $d$ non-zeros, and $F(d)$ is the probability that a row has $d$ or more non-zeros.  In that case, drawing rows identically and independently (as is needed for a sound statistical analysis) implies that we expect in a sample on average a fraction $F(d)$ of rows to require $d$ or more non-zeros, and the standard deviation on this estimate is $\sqrt{F(d) (1-F(d))/n}$ where $n$ is the total number of rows.  We can from this compute the probability $\delta(\lambda)$ that the real fraction of rows with more than $d$ non-zeros is larger than $F(d)+\lambda$ for $\lambda>0$.

In case one doesn't know the population distribution but has a sample from which one can estimate these statistics, a similar approach can be followed applying the appropriate formulas, e.g., the standard deviation of a sample of size $s$ is $\sqrt{p(1-p)/(s-1)}$ with $p$ the measured probability.

\begin{figure}
	\centering
	\begin{subfigure}{.45\linewidth}
		\centering
		\includegraphics[width=\linewidth]{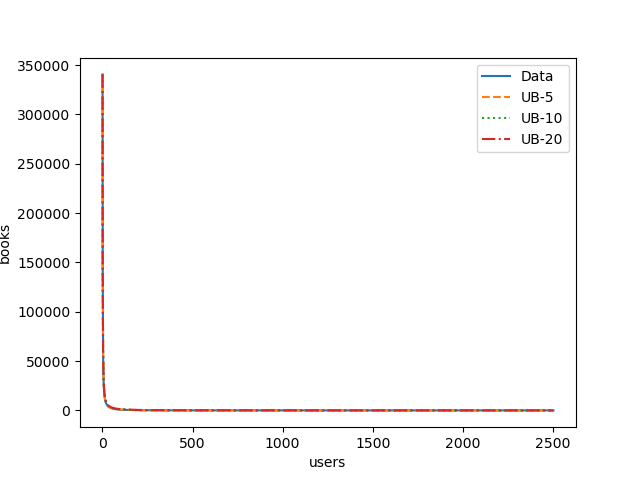}
		\caption{Linear scale}
		\Description[Plot showing the distribution of row non-zero counts, and upper bounds based on knowledge of the population distribution]{Plot showing the distribution of row non-zero counts, and upper bounds based on knowledge of the population distribution}
		\label{fig:ub.lin}
	\end{subfigure}
	\begin{subfigure}{.45\linewidth}
		\centering
		\includegraphics[width=\linewidth]{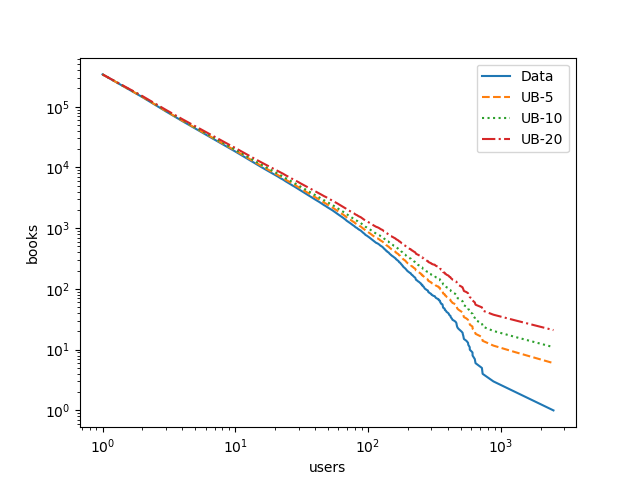}
		\Description[Plot showing the distribution of row non-zero counts, and upper bounds based on knowledge of the population distribution]{Plot showing the distribution of row non-zero counts, and upper bounds based on knowledge of the population distribution}
		\caption{Logarithmic scale}
		\label{fig:ub.log}
	\end{subfigure}
	\caption{Distribution of row non-zero counts, and upper bounds based on knowledge of the population distribution}
\end{figure}

Fig. \ref{fig:ub.lin} plots the number of non-zeros in the template matrix for every row in a way similar to Fig. \ref{fig:dp.recomm} for the Bookcrossing dataset.  As there is no differential privacy cost, we can take the size of a template block to be a single row.  The several curves use 5, 10 and 20 standard deviations, i.e., going up to a high level of confidence that the real data will fit in the template matrix if the data is drawn according to the population distribution (or the same distributon as the data from which statistics were taken).  In the plot, all curves coincide because much less margin is needed compared to the differential privacy based strategy.  To see the difference better, the same plot is shown in Fig. \ref{fig:ub.log} in logarithmic scale.

\section{Adapting existing secure sparse multiplication to the outsourced setting}
\label{app:adapt-non-outsourced}
Our work introduces new protocols to perform secure sparse multiplications in an outsourced setting.
While we introduce the first outsourced secure sparse multiplications, existing works \cite{chen_when_2021,cui_exploiting_2021,schoppmann_make_2019} had already described secure protocols, but in a non-outsourced setting.
Instead of introducing new protocols, one may wonder whether the existing protocols can be adapted to the outsourced setting.
This appendix details why such an adaptation of \cite{chen_when_2021,cui_exploiting_2021,schoppmann_make_2019} is not straightforward and can result in inefficient constructions.

\paragraph{Adapting \cite{chen_when_2021,cui_exploiting_2021}}
These works described sparse-dense matrix multiplications where one party $P_1$ knows the plaintext sparse matrix $X$ and the other party $P_2$ knows the plaintext dense matrix $Y$.
In both protocols, the dense matrix $Y$ is homomorphically encrypted and the first party $P_1$ performs a dense multiplication between their plaintext sparse matrix $X$ and the encrypted dense matrix $Y$ it receives from $P_2$.
In \cite{chen_when_2021}, the encrypted matrix $Y$ is entirely transmitted to the first party.
In \cite{cui_exploiting_2021}, the parties use a private information retrieval to avoid transmitting the whole matrix.

Despite their differences, these two protocols cannot be adapted to the outsourced setting for the same reasons.
First, the homomorphic encryption step is complex in an outsourced setting because there is no evident key dealer (contrary to the two-data-owner setup).
Second, these protocols are possible because $P_1$  knows the position of all non-zero values in $X$.
Indeed, $P_1$ performs a multiplication between the plaintext sparse matrix and the encrypted dense matrix.
This operation is possible because homomorphically encrypted values can be trivially multiplied by plaintext constants.
If these non-zero positions are hidden, this simple multiplication is no longer possible.
There is then no straightforward solution to adapt the approach of \cite{chen_when_2021,cui_exploiting_2021} to the outsourced setting.

Third, these protocols only support sparse-dense multiplications: the matrix $Y$ \textbf{must be dense}, or at least be converted into a dense representation before being homomorphically encrypted.
Otherwise, $P_1$ may learn which values of $Y$ are zero.
As a result, the memory needed to store $Y$ could blow up with several orders of magnitude, and so could the communication and computation costs.

\paragraph{Adapting Schoppmann et al. \cite{schoppmann_make_2019}}
They adopted a different approach than \cite{chen_when_2021,cui_exploiting_2021}: they introduce a low-level protocol called ``ROOM'' (Read-Only Oblivious Map) upon which they built their matrix multiplications.
The ROOM protocol performs a secure look-up in a secret key-value store.
This protocol guarantees that neither the key-value store, nor the look-up query are revealed.
They introduce three instantiations of the ROOM protocol: BasicROOM, CircuitROOM, and PolyROOM

They presented all their functionality in a two-data-owner setup with the query issued by the first party and the key-value store owned by the second.
However, we can easily write these functionalities for an outsourced setup.

Out of the three ROOM instanciations, only CircuitROOM is easily adaptabled to the outsourced setting because it relies on oblivious sorting and oblivious shuffling.
Since our protocols also rely on sorting and shuffling, this observation may lead the reader to two questions: (1) What are the differences between the CircuitROOM-based sparse multiplications and our protocols? (2) What are the implications of these differences on efficiency?

They use the ROOM protocol to extract a dense sub-matrix from a sparse matrix.
In contrast, we use sorting and shuffling for the overall sparse matrix multiplication.
Besides this key algorithmic difference, \emph{efficient} ROOM-based protocols in the outsourced setting is not straightforward.
The following paragraphs study each matrix multiplication types to identify the algorithmic and efficiency differences.

To multiply two sparse vectors $x$ and $y$, Schoppmann et al. \cite{schoppmann_make_2019} first use CircuitROOM to extract a dense vector of size $\nnz(x)$ from $y$ (i.e., corresponding to the non-zero values of $x$).
This operation costs $O((\nnz(x)+\nnz(y))\log(\nnz(x)+\nnz(y)))$ communications.
Their protocol then performs a dense vector multiplication between the non-zero values of $x$ and the extracted dense vector from $y$.
This overall communication cost is $O((\nnz(x)+\nnz(y))\log(\nnz(x)+\nnz(y)))$, as in our sparse vector multiplication (Alg. \ref{alg:sparse_vect_mult}).
However, their protocol has larger constant factors: while our protocol only requires an oblivious sort, CircuitROOM performs both an oblivious sort and oblivious shuffle.
Their protocol even requires a dense vector multiplication after CircuitROOM.

To multiply a sparse matrix $X$ with a sparse vector $y$, Schoppmann et al. \cite{schoppmann_make_2019} call the ROOM protocol once for each data owner.
In their protocol, each data owner shares a dense sub-matrix containing its non-zero columns (i.e., a column with at least one non-zero value).
Hence, their protocol call the matrix-vector multiplication (and the underlying ROOM protocol) for each matrix of non-zero columns; in other words, for each data owner.
On the contrary, our protocol can group the non-zeros from all data owners and process them all at the same time.

In their paper, this design choice is not an issue because there is only two data owners.
However, in an outsourced setup, each matrix row can be owned by a different data owner.
Hence, the matrix multiplication would perform one vector multiplication per matrix row.
Their communication cost is then $O((\nnz(X)+n\cdot\nnz(y))\log(\nnz(X)+\nnz(y)))$,  which is \textbf{worse than our cost, by a linear factor} $n$.
Indeed, Sec. \ref{subsec:alg-mat-vec} presents a non-trivial matrix-vector multiplication avoiding this linear factor (i.e., Alg. \ref{alg:sparse_mat_vect_mult}).
Since ML datasets can contain thousands or even millions of rows (e.g., 340K in the Bookcrossing dataset), the CircuitROOM-based matrix-vector multiplication would be inefficient compared to our Alg. \ref{alg:sparse_mat_vect_mult} (and even compared to the dense multiplication).

To compute a correlation matrix $X^{\top}X$, Schoppmann et al. \cite{schoppmann_make_2019} require, in addition to ROOM, another sub-protocol called "Scatter".
Unfortunately, their protocol described in the non-outsourced setting cannot be adapted easily to an outsourced setting because it relies on oblivious transfer, a purely two-party protocol with a party having access to the plaintext secret.
Hence, we cannot adapt their sparse matrix-matrix multiplication to the outsourced setting.

To sum up, the adaptation of the CircuitROOM-based sparse multiplication is either inefficient (for vector-vector and matrix-vector multiplications) or not straightforward (for matrix-matrix multiplications).

\paragraph{Adapting Lodia \cite{yu_lodia_2025}}
Finally, Yu et al. \cite{yu_lodia_2025} proposed a sparse matrix multiplication based on fully homomorphic encryption (FHE); assuming that a single party owns the sparse matrix.
This data owner decomposes the large sparse matrix into a multiplication of smaller and denser matrices.
These smaller matrices are then encrypted and can be multiplied by a third party securely.
This matrix decomposition could be used also for MPC protocols, as a sparse matrix encoding.
Unfortunately, it makes a strong assumption: having a single data owner; which differs from our outsourcing setting in which each row can be owned by a different party.
The matrix decomposition being a complex non-linear operation, there is no straightforward way to implement it under MPC and avoid the single-owner assumption.

\end{document}